\documentclass[pra,aps,amsmath,amssymb,superscriptaddress,twocolumn,longbibliography]{revtex4-1}
\usepackage{graphicx,multirow}
\usepackage{color,outlines}
\usepackage[english]{babel}
\usepackage{amsthm}
\usepackage[sc,osf]{mathpazo}
\usepackage{hyperref}
\usepackage{physics}
\usepackage{outlines}

\usepackage{adjustbox}

\usepackage[a4paper,margin=1in]{geometry}
\usepackage{amsmath,amssymb,amsfonts}
\usepackage{mathtools}
\usepackage{bm}
\usepackage{physics}
\usepackage{graphicx}
\usepackage{hyperref}

\usepackage{tabularx} 
\usepackage{array}
\usepackage{tikz}
\usepackage{adjustbox}
\newcolumntype{C}{>{\centering\arraybackslash}X}

\usepackage[sc,osf]{mathpazo}
\usetikzlibrary {positioning}

\definecolor{green2}{RGB}{0,100,0}

\newtheorem{theorem}{Theorem}

\newtheorem{theoremD}{Theorem}
\newtheorem{definition}[theoremD]{Definition}

\newtheorem{theoremR}{Theorem}
\newtheorem{remark}[theoremR]{Remark}

\newtheorem{theoremP}{Theorem}
\newtheorem{proposition}[theoremP]{Proposition}

\begin{document}

\title{Entangling Power  Dynamics:  Ergodicity and Mixing}

\author{Rahul V}
\email{rahulsomasundar@gmail.com}
\author{Ranjan Modak}
\email{ranjan@iittp.ac.in}
\author{Shaon Sahoo}
\email{shaon@iittp.ac.in}
\author{S. Aravinda }
\email{aravinda@iittp.ac.in}

\affiliation{Department of Physics, Indian Institute of Technology Tirupati, Tirupati, India~517619}  
\begin{abstract}
We study quantum dynamics through the lens of entanglement generation and characterize the underlying unitary evolution by the distinct signatures it imprints on the time-dependent entangling power.
For a unitary operator, we characterize ergodicity by the equality between its long-time-averaged entangling power and the Haar-averaged linear entropy. We define mixing more stringently as the convergence of the time-dependent entangling power itself to the Haar value at long times. Within this framework, we establish the ergodic hierarchy of dynamical behavior, showing in particular that mixing implies ergodicity, whereas ergodicity does not necessarily imply mixing. As an application, we find that two-qubit unitary gates are neither ergodic nor mixing: their long-time-averaged entangling power can take only four discrete values, none of which coincides with the Haar average. We then investigate many-body dynamics using the kicked Ising chain and find that the long-time-averaged entangling power converges to the Haar value in both integrable and nonintegrable cases, indicating ergodicity. Remarkably, however, the nonintegrable chain exhibits mixing, whereas the integrable chain, despite being ergodic, is demonstrably nonmixing. 
We also introduce a Lyapunov-like exponent to characterize the rate at which the time-dependent entangling power approaches its saturation value. We find that this exponent increases systematically with the degree of integrability breaking in the many-body system. 
Our results establish entanglement generation as a useful framework for characterizing dynamical systems and reveal qualitatively different signatures of integrability beyond conventional diagnostics.

\end{abstract}

\maketitle

\section{Introduction}

Quantum entanglement is a key feature of quantum mechanics and is a main resource for quantum information, quantum computation, and quantum communication \cite{ArnoldAvez1968,Gisin2007,Nielsen2010}. Beyond quantum information science, entanglement has emerged as a powerful probe of many-body quantum systems, providing valuable insights into quantum phase transitions \cite{Amico2008,Osterloh2002,Osborne2002}, thermalization \cite{Deutsch1991,Srednicki1994,DAlessio2016}, information scrambling \cite{Hayden2007,Sekino2008,Swingle2018}, and nonequilibrium quantum dynamics \cite{Calabrese2005,Calabrese2007,Eisert2015}. In isolated many-body systems, unitary dynamics can generate entanglement, which can be studied using measures of entanglement \cite{zanardi2000entangling,zanardi2001entanglement}. Consequently, the evolution of entanglement has become a tool for investigating integrability, quantum chaos, thermalization, and ergodicity in quantum many-body systems \cite{Calabrese2005,Calabrese2007,Eisert2015,Pal2018}.

Ergodicity in a quantum system refers to the property that its long-time dynamics explores the accessible Hilbert space in accordance with the predictions of statistical ensembles, leading to the relaxation of local observables toward their thermal equilibrium values \cite{Deutsch1991,Srednicki1994,DAlessio2016}. In classical ergodic theory, ergodicity is one of the fundamental statistical properties of a probability-measure-preserving dynamical system. It forms part of the ergodic hierarchy, with mixing and Bernoulli properties representing progressively stronger forms of chaotic behavior \cite{Cornfeld1982}. Mixing represents a stronger notion of ergodicity, in which temporal correlations decay over time, causing the system to lose memory of its initial conditions, such that states widely separated in time approach statistical values \cite{Zaslavsky1981,Peres1984}.

Dual-unitary circuits are quantum circuits in which the local two-site gate is unitary both in the usual temporal direction and after a spatial reshuffling of its indices, making the gate unitary in both space and time \cite{Bertini2019}. This special structure provides an exactly tractable setting for studying the quantum ergodic hierarchy, including ergodicity, mixing, and the stronger Bernoulli property, and has been further generalized to triunitary circuits and hierarchical extensions of dual unitarity solvable models
\cite{Aravinda2021,Rather2020,Jonay2021,yu2024hierarchical}. Within quantum many-body physics, these ideas are deeply interconnected with the Eigenstate Thermalization Hypothesis, quantum chaos, information scrambling, and operator spreading, which together constitute the contemporary theoretical framework for thermalization in isolated quantum systems \cite{Deutsch1991,Srednicki1994,DAlessio2016,Jisha_2024,MISHRA2025131000}. A study of the ergodic and mixing nature of quantum channels is reported in \cite{PhysRevA.110.042607,Burgarth2013,Movassagh2022,Movassagh2021}.

Entangling power quantifies the ability of a unitary operator to generate entanglement on average from an initial product state \cite{Zanardi2000,Kraus2001,Makhlin2002,Zhang2003}. Iterative applications of a unitary with local dynamics drive its entangling power exponentially toward the corresponding random-unitary value \cite{Jonnadula2017}. For many-body time-evolution operators, the long-time entangling power provides a distinction between different dynamical regimes, with integrable and nonintegrable systems exhibiting different saturation behavior \cite{Pal2018}. Recent studies have generalized entangling power to multipartite and dual-unitary systems, revealing richer entanglement dynamics beyond individual bipartite gates \cite{Manna2024,Qiu2025,Malik2026}. Entangling power has made major contributions to pseudo-random matrix theory, the realization of high-fidelity entanglers, quantum-classical correspondence in quantum channels, SU(2)-symmetric dissipative quantum many-body dynamics, and imperfect entangling power of quantum gates \cite{PRXQuantum.6.020322,npr7-b7kq,PhysRevE.111.014210,d8kg-h1t7,jl2b-bxfn}. These developments raise the question of whether entangling power and its long-time averages can be used to characterize the ergodic, mixing, and chaotic properties of quantum evolution.

We therefore begin with the simple two-qubit system as a convenient setting to study these properties. It is possible to represent all the nonlocal and local two-qubit gates in the Weyl chamber using the Cartan decomposition of the SU(4) group by parameterizing three variables \cite{Zhang2003}. The long-time behavior under successive applications of a two-qubit gate was investigated, where arbitrary iterations of the gate were related to the dynamics of a free particle in a three-dimensional billiard associated with the Weyl chamber \cite{Mandarino2018}. This geometric description provides a useful picture of the trajectories generated by repeated applications of the unitary and forms a natural starting point for studying the long-time and time-averaged behavior of the entangling power.

We then extend our study to many-body systems, focusing on quantum spin chains as a natural setting for investigating the interplay between entangling power, ergodicity, and mixing. Quantum spin chains are great systems for studying these ideas because they show a rich mix of particle interactions, quantum correlations, and out-of-equilibrium behavior. Because of this, physicists use tools like entanglement growth, entanglement entropy, fidelity, operator entanglement, and entangling power to study quantum phase transitions, topological phases, Floquet (periodically-driven) dynamics, and ergodic behavior in low-dimensional systems \cite{Amico2008,Sahoo2013,Sahoo2014,Sahoo2020}.

The structure of the paper follows that first section contains formalism and its relation to the Ergodic Hierarchy and Chaos of the multipartite system are presented in Sec.~\ref{sec_1}. We then apply this framework to Cartan-decomposed two-qubit gates and analyze their behavior in the Weyl chamber in Sec.~\ref{sec_2}. Finally, we extend the analysis to the kicked spin Ising chain and investigate the behavior of the average entangling power in this many-body system in Sec.~\ref{sec_3}.

\section{Ergodic Hierarchy of Entangling Power}\label{sec:formalism}
\label{sec_1}

In this section, we discuss the dynamics of entangling power and introduce measures to characterize its temporal behavior. 
In the following, $U$ represents a unitary evolution operator which acts on the Hilbert space $\mathcal{H}_A\otimes\mathcal{H}_B$ of a bipartite system $AB$. The time-evolution by the unitary operator is represented by the sequence $\{I, U, U^2, \cdots\}$. We start with the definition of the entangling power of the operator $U$.

\begin{definition}[Entangling Power $e_P(U)$ \cite{zanardi2000entangling}]
 Consider $U:\mathcal{H}_A\otimes\mathcal{H}_B \rightarrow \mathcal{H}_A\otimes\mathcal{H}_B$ and $|\psi_A\rangle\otimes|\psi_B\rangle \in \mathcal{H}_A\otimes\mathcal{H}_B$. The entangling power of $U$ is defined as
\begin{equation}
e_p(U)=
\int d\psi_A\, d\psi_B\;
\mathcal{E}\!\left(
U\left(|\psi_A\rangle\otimes|\psi_B\rangle\right)
\right),
\end{equation}
where $d\psi_A$ and $d\psi_B$ denote the Haar measures on the pure-state manifolds of the subsystems, and $\mathcal{E}(\cdot)$ is an appropriate measure of bipartite entanglement.
    \label{df_ep}
\end{definition}

For operational convenience, it is common to choose the linear entropy for $\mathcal{E}(\cdot)$ \cite{Zanardi2000,Zanardi2001,Wang2003,Pal2018}.   
For a pure state $\ket{\psi_{AB}} \in \mathcal{H}_A\otimes\mathcal{H}_B$, the linear bipartite entanglement entropy is define as,
\begin{equation}
    \mathcal{E}(|\Psi_{AB}\rangle)
= 1-\operatorname{Tr}(\rho_A^2), 
\label{df_LE}
\end{equation}
where
$\rho_A=\operatorname{Tr}_B
\left(
|\Psi_{AB}\rangle\langle\Psi_{AB}|
\right)$ is the reduced density matrix of the subsystem $A$. In fact, for a bipartite system with $dim(\mathcal{H}_A)=dim(\mathcal{H}_B)=d$, the entangling power in terms of linear entropy can be given in a closed form \cite{zanardi2001entanglement}:
\begin{equation}
e_p(U) =\frac{d^2}{(d+1)^2} \left[E(U)+E(US)-E(S)\right],
\label{eq_ep_closed}
\end{equation}
where $S$ is the SWAP operator: $S|\psi_A\rangle\otimes|\psi_B\rangle=|\psi_B\rangle\otimes|\psi_A\rangle$. Here $E(U)$ is the operator entanglement \cite{Zanardi2001} which is defined as 
\begin{equation}
E(U)
=
1-\frac{1}{d^{4}}
\operatorname{tr}
\left[
\left(U^{R_1}U^{R_1\dagger}\right)^2
\right]
\end{equation} where
\begin{equation}
\langle \beta \alpha | U^{R_1} | j i \rangle
=\langle i \alpha | U | j \beta \rangle.
\end{equation}

We now consider the entropy of a Haar-randomly chosen state, i.e., the linear entropy averaged over states from the Hilbert space $\mathcal{H}_A\otimes\mathcal{H}_B$. The general result is presented below:

\begin{theorem}[Average linear entropy $\langle \mathcal{E}_{m,n}\rangle$]
Let $dim(\mathcal{H}_A)=m$ and $dim(\mathcal{H}_B)=n$ with $m\le n$. For Haar-randomly chosen $\ket{\psi}\in \mathcal{H}_A\otimes\mathcal{H}_B$, we have 
\begin{equation}
    \langle \mathcal{E}_{m,n}\rangle=\int d\psi \mathcal{E}(\ket{\psi})=\frac{(m-1)(n-1)}{mn+1},
\label{eq_avle} 
\end{equation}
where $d\psi$ denotes the Haar measure on the pure-state manifold of the total system.
\label{lmm_avE}
\end{theorem}
This result is a direct consequence of the definition of the linear entropy (Eq. \ref{df_LE}) and Lubkin's result  \cite{Lubkin1978},
\begin{equation}
\left\langle \mathrm{Tr}(\rho_A^2)\right\rangle
= \frac{m+n}{mn+1}.
\end{equation}

Sometimes, it is more convenient to use the normalized linear entropy, 
\begin{equation}
{\mathcal{E}}(\ket{\psi})=\frac{m}{m-1}(1-\mathrm{Tr}(\rho_A^2)).
\label{eq_nle}
\end{equation}
This normalization ensures that ${\mathcal{E}} = 1$ for the maximally mixed state $\rho_A$. With this normalized linear entropy and for $m=n=d$, we have the following average linear entropy (cf. Eq. \ref{eq_avle}): 
\begin{equation}
    \langle {\mathcal{E}} \rangle=\frac{d(d-1)}{d^2+1}.
    \label{eq_avnle}
\end{equation}

We now turn our attention to the dynamics and first define the time-averaged entangling power over a finite number of evolution steps. In the following, $e_p(U^k)$ denotes the average linear entropy generated by the unitary $U$ after $k$ steps of evolution.
 
\begin{definition} [Time-averaged $e_p(U)$ after finite evolution]
Starting from a Haar-random product state, the average amount of entropy produced during a time evolution of $n$ steps is 
\begin{equation}
    \overline{e_p}(U^n)=\frac{1}{n+1}\left[\sum_{k=0}^n e_p(U^k)\right].
    \label{eq_fte}
\end{equation}
\label{df_fte}
\end{definition}
This definition can also be extended for long-time average:
\begin{equation}
    \overline{e_p}^{sat}(U)=\lim_{n\to \infty}\overline{e_p}(U^n).
    \label{eq_lte}
\end{equation}

\begin{proposition}
The saturation value $\overline{e_p}^{sat}$ is a well defined quantity, i.e., the limit in Eq. \ref{eq_lte} exists.
\label{prop1}
\end{proposition}
\begin{proof}
We note that $\delta_k=\left|\overline{e_p}(U^{k+1})-\overline{e_p}(U^{k})\right|=\left|(\frac{k+1}{k+2}-1)\overline{e_p}(U^{k})+\frac{e_p(U^{k+1})}{k+2}\right|$. Now, since $0\le e_p(U^{k})\le 1$ and $ 0\le \overline{e_p}(U^{k}) \le 1$ for each $k$, we have $\lim_{k\to \infty} \delta_k =0$.
\end{proof}

The long-time average of the entangling power is a well-defined quantity, and for each unitary operator $U$, we will have one unique saturation value $\overline{e_p}^{sat}(U)$. But in general, there can be multiple such saturation values, called fixed points, for a quantum system of given Hilbert space. As will be seen in Section \ref{sec_2}, each fixed point represents a class of unitary operators. Each operator from a class corresponds to the same fixed point.
\begin{definition}[Entropic fixed point]
For a given unitary operator $U$, the finite-time averaged quantity $\overline{e_p}(U^n)$ approaches a saturation value as $n \to \infty$. This limiting value, denoted as $\overline{e_p}^{sat}(U)$, is an entropic fixed point of the evolution generated by $U$.
\end{definition}

We are now in a position to discuss the different categories of dynamical systems within the present framework.

A Bernoulli system forgets its past every single moment during its dynamics, i.e., it does not even have short-term memory. In our formalism, we define it below.
\begin{definition}[A Bernoulli system]
A system described by the evolution operator $U$ is called Bernoulli if $ e_P(U^n) := \langle \mathcal{E} \rangle ~\forall n\ge 1$.
\label{df_berno}
\end{definition}

A system with the mixing property may carry memory of its past, but the correlation decays with time. In the current formalism, we have:

\begin{definition}[A  mixing operator]
    An operator $U$ describing the evolution of a system is called mixing if $\lim_{n\to \infty} e_P(U^n) := \langle \mathcal{E} \rangle$.
\label{df_mixing}
\end{definition}
In a similar way, one can also define weak mixing operator for which the correlation decays in an average sense. Following the ergodic hierarchy, we finally define an ergodic system as:

\begin{definition}[An ergodic system]
    A system described by the evolution operator $U$ is called ergodic if $\overline{e_p}^{sat}(U):=\langle {\mathcal{E}}
    \rangle$.
\label{df_ergodic}
\end{definition}

\begin{remark} 
Since the definition of entangling power involves product states, in the present formalism of characterizing the dynamical system, the evolution is understood to begin from a product state. Thus, the initial state for the dynamics is chosen as a Haar-random state from the manifold of zero-entanglement (zero-entropy) states.

This restriction on the initial state can be somewhat relaxed by shifting the origin of time. However, this is not possible for a Bernoulli system with small $n$ (see Definition \ref{df_berno}).
\end{remark}

Next, we prove the ergodic hierarchy between different dynamical systems within our formalism. 
\begin{proposition}
    A Bernoulli operator $U$ is also mixing.
\end{proposition}
\begin{proof}
    The Definition \ref{df_berno} of a Bernoulli operator implies that a system is Bernoulli if $ e_P(U^n) := \langle {\mathcal{E}} \rangle ~\forall n\ge N$ for any fixed $N$. Taking $n \to \infty$ proves the proposition. 
\end{proof}

\begin{proposition}
     A mixing operator $U$ is also ergodic.
\end{proposition}
\begin{proof}
A mixing operator $U$, as in Definition \ref{df_mixing}, can be redefined in the following operational form. The operator $U$ is called a mixing, if, for any $\epsilon$, we can find a number $N$ so that 
\begin{equation}
|e_p(U^n)-\langle {\mathcal{E}} \rangle|<\epsilon ~~~ \forall n\ge N. 
\label{df2_mixing}
\end{equation}

Using the explicit expression of $\overline{e_p}^{sat}$, from Eqs. \ref{eq_fte} and \ref{eq_lte}, we get\\

$|\overline{e_p}^{sat}(U)-\langle {\mathcal{E}} \rangle|$
\begin{equation*}
    \begin{split}
     &=\lim_{m\to \infty} \frac{1}{m+1}\abs{\sum_{n=0}^m e_p(U^n) - (m+1)\langle {\mathcal{E}} \rangle }\\
     &=\lim_{m\to \infty} \frac{1}{m+1}\left| \sum_{n=0}^N 
     (e_p(U^n) - \langle {\mathcal{E}} \rangle) \right.\\
     &\left. ~~~~~~~~~~~~~~~~~~~~~~~ + \sum_{n=N+1}^m 
     (e_p(U^n) - \langle {\mathcal{E}} \rangle) \right|.
    \end{split}
\end{equation*}

Taking $N \ll m$, we obtain
\[
|\overline{e_p}^{sat}(U)-\langle {\mathcal{E}} \rangle|
\le
\lim_{m\to \infty}
\frac{1}{m+1}
\sum_{n=N+1}^m
\left|e_p(U^n)-\langle {\mathcal{E}} \rangle\right|.
\]
We note that the upper bound of the ignored term is 
\[
\lim_{m\to \infty}
\frac{1}{m+1} \sum_{n=0}^N \left| e_p(U^n)-\langle {\mathcal{E}} \rangle \right|\le \lim_{m\to \infty}
\frac{N+1}{m+1} = 0.
\]

Now if the system is mixing, then following Eq.~\ref{df2_mixing}, we get
\[
|\overline{e_p}^{sat}(U)-\langle {\mathcal{E}} \rangle|
<
\lim_{m\to \infty}
\frac{m-N-1}{m+1}\epsilon
=
\epsilon.
\]
This proves our proposition that mixing operators are also ergodic.
\end{proof}

Although, for a mixing system, $e_p(U^n)$ eventually approaches the saturation value $\langle {\mathcal{E}} \rangle$, it is also of interest to characterize how this convergence occurs with time. For this purpose, we define a Lyapunov-like exponent.

\begin{definition}[Lyapunov-like exponent]
For a mixing system, the time-dependent entangling power approaches saturation in the following way: 
\begin{equation}
|\langle {\mathcal{E}} \rangle- {e_p}{(U^n)}|\simeq \langle {\mathcal{E}} \rangle e^{-\lambda n}, 
\end{equation} 
for large $n$. Here $\lambda$ is the Lyapunov-like exponent. 
\end{definition}
Above form is motivated from the RMT results, as discussed in Section \ref{sec_3}. This exponent $\lambda$ characterizes the inverse timescale over which the system approaches the state of maximal mixing accessible under the dynamics generated by $U$.

\section{Two qubit unitary operators}
\label{sec_2}
We next consider two-qubit gates within our framework; for that, the Cartan decomposition of the Lie group $SU(4)$ is the best choice because it can neatly capture the nonlocal feature of the two-qubit gates.  Any two-qubit unitary lives in
the fifteen-dimensional group $SU(4)$, and its generators split
naturally into a local part  and a non-local part . Because local operations alone cannot generate or alter
entanglement, it is useful to isolate the non-local content from the
rest of the dynamics.The Cartan theorem\cite{Makhlin2002,Zhang2003}, states that any $U' \in SU(4)$ can be written as
\begin{equation}
U' = (u_{1}\otimes u_{2})\, U \,(u_{3}\otimes u_{4}),
\label{eq:cartan_re}
\end{equation}
where each $u_{i}\in SU(2)$ $(i=1,\ldots,4)$ is a single-qubit
rotation, and
\begin{equation}
U=\exp\left[-i\left(c_{1}\,\sigma_{x}\otimes\sigma_{x}
+c_{2}\,\sigma_{y}\otimes\sigma_{y}
+c_{3}\,\sigma_{z}\otimes\sigma_{z}\right)\right]
\label{eq:canonical_re}
\end{equation}
is the non-local operator.We can express the whole expression of the entangling power of two-qubitubit gate by $(c_{1},c_{2},c_{3})$, known as the Cartan coordinates.The main generators of the two qubit gates $\sigma_{x}\otimes\sigma_{x}$,
$\sigma_{y}\otimes\sigma_{y}$, $\sigma_{z}\otimes\sigma_{z}$ generate a
maximal commuting subalgebra --- they form the Cartan
subalgebra of the non-local sector of $\mathfrak{su}(4)$:
\begin{equation}
[\sigma_{\alpha}\otimes\sigma_{\alpha},\,
\sigma_{\beta}\otimes\sigma_{\beta}]=0,
\qquad \alpha,\beta\in\{x,y,z\}.
\end{equation}
  \cite{Kraus2001}, these commuting operators are simultaneously diagonal in the magic basis.

$$
\begin{array}{ll}
|\phi_1\rangle = \dfrac{1}{\sqrt{2}}\left(|00\rangle + |11\rangle\right),
&
|\phi_2\rangle = \dfrac{i}{\sqrt{2}}\left(|00\rangle - |11\rangle\right),
\\[6pt]
|\phi_3\rangle = \dfrac{i}{\sqrt{2}}\left(|01\rangle + |10\rangle\right),
&
|\phi_4\rangle = \dfrac{1}{\sqrt{2}}\left(|01\rangle - |10\rangle\right).
\end{array}
$$
which consists of the Bell states up to global phases. Consequently, the nonlocal operator $A$ is diagonal in this basis,
\begin{equation}
U=\sum_{k=1}^{4}e^{-i\gamma_k}
|\Phi_k\rangle\langle\Phi_k|,
\end{equation}
where the eigenphases are
\begin{align}
\gamma_1 &= c_1-c_2+c_3,\nonumber\\
\gamma_2 &= -c_1+c_2+c_3,\nonumber\\
\gamma_3 &= -c_1-c_2-c_3,\nonumber\\
\gamma_4 &= c_1+c_2-c_3.
\end{align}
Since the eigenvectors remain unchanged under repeated application, the
$n$th power simply multiplies each eigenphase by $n$,
\begin{equation}
\begin{split}
U^{n}
&=
\sum_{k=1}^{4}
e^{-in\gamma_k}
|\Phi_k\rangle\langle\Phi_k| \\
&=
\exp\!\left[
-in\left(
c_{1}\sigma_x\otimes\sigma_x
+c_{2}\sigma_y\otimes\sigma_y
+c_{3}\sigma_z\otimes\sigma_z
\right)
\right].
\end{split}
\end{equation}
Therefore, the Cartan coordinates evolve linearly under iteration,
$
(c_1,c_2,c_3)\longrightarrow (nc_1,nc_2,nc_3),
$
with the understanding that they are subsequently folded back into the Weyl chamber using the canonical equivalence relations of two-qubit gates. This three-parameter description is, however, not one-to-one: several
different triples $(c_{1},c_{2},c_{3})$ can represent gates that are
equivalent under local operations, owing to permutations of the
coordinates, sign changes arising from local basis choices, and the
discrete Weyl-group symmetries inherent to $SU(4)$. Once these
redundancies are removed, every local-equivalence class has exactly
one representative lying inside a fundamental region called the Weyl
chamber, a tetrahedral domain defined by
$
0\le c_{3}\le c_{2}\le c_{1}\le \frac{\pi}{4}.
\label{eq:weyl_re}
$
Hence, each point of this chamber corresponds to exactly one two-qubit gate, up to local equivalence.
For a canonical gate $U$ with Weyl-chamber coordinates
$(c_{1},c_{2},c_{3})$, the entangling power of the $n$-fold
application $U^{n}$ is
\begin{equation}
e_{p}(U^{n})
=\frac{1}{6}\Big[3-\big(C_{12}(n)+C_{23}(n)+C_{31}(n)\big)\Big],
\label{eq:ep_Un}
\end{equation}
where\cite{Balakrishnan2010}
\begin{equation*}
C_{ij}(n)\equiv\cos(2nc_{i})\cos(2nc_{j}).
\end{equation*}
For two-qubit pure states ($m=n=2$), \ref{eq_avnle}  becomes $\langle  \mathcal{E}\rangle=2/5=0.4$. ,
which serves as the random-state benchmark for bipartite entanglement.
We know that
\begin{equation}
D_N(\theta)
=\sum_{n=0}^{N}\cos(n\theta)
=
\frac{\sin[(N+1)\theta/2]\cos(N\theta/2)}
{\sin(\theta/2)}.
\label{cos}
\end{equation}
We can extend \ref{eq:ep_Un} using \ref{cos} to find the quantity 

\begin{align}
\overline{e_{p}}(U^n)
&=
\frac12
-\frac{1}{12(n+1)}
\sum_{(ij)}
\Big[
D_n\!\left(2(c_i-c_j)\right)
\nonumber\\
&\hspace{2.8cm}
+
D_n\!\left(2(c_i+c_j)\right)
\Big],
\label{eq:ep_finite}
\end{align}
and the summation over the pairs
$(ij)=(12),(23),(31)$.
From the first look it is evident that whenever N tends to infinity the average entangling power tends to 0.5 ie
\begin{equation}
\lim_{N\rightarrow\infty}
\overline{e_{p}}(U^n)
=\frac12.
\label{eq:ep_generic_limit}
\end{equation}
The exceptional cases arise when one or more resonance conditions
$c_i\pm c_j=k\pi$ ($k\in\mathbb{Z}$) are satisfied. 
\begin{equation}
F(\theta)
=
\lim_{n\rightarrow\infty}
\frac{1}{n+1}
\sum_{k=0}^{n}
\cos(k\theta)
=
\delta_{\theta},
\end{equation}
where
\[
\delta_{\theta}
=
\begin{cases}
1, & \theta=2\pi k,\\
0, & \text{otherwise},
\end{cases}
\qquad
k\in\mathbb{Z},
\]
one obtains
\begin{equation}
\big\langle C_{ij}\big\rangle
=
\frac12
\left(
\delta_{c_i-c_j}
+
\delta_{c_i+c_j}
\right),
\label{eq:Cij_average}
\end{equation}
where $\delta_x=1$ if $x=k\pi$ and $0$ otherwise. Substituting
Eq.~\eqref{eq:Cij_average} into the long-time average of
Eq.~\eqref{eq:ep_Un} gives
infinite-time averaged entangling power can therefore be written as
\begin{equation}
\begin{split}
\lim_{n\rightarrow\infty}\overline{e_p}(U^n)
&=
\frac{1}{2}
-
\frac{1}{12}
\sum_{(ij)}
\left[
\delta_{c_i-c_j}
+
\delta_{c_i+c_j}
\right] \\
&=
\frac{1}{2}
-
\frac{m}{12},
\label{eq:ep_infinite}
\end{split}
\end{equation}
where $
m
=\sum_{(ij)}
\left[
\delta_{c_i-c_j}
+
\delta_{c_i+c_j}
\right],
$
with $(ij)=(12),(23),(31)$, denotes the total number of surviving
resonance contributions. Here, $\delta_x=1$ for $x=0$ and vanishes
otherwise.The ordering of the Weyl chamber,
\begin{equation}
0\leq c_3\leq c_2\leq c_1\leq\frac{\pi}{4},
\end{equation}
strongly restricts the possible resonance patterns. In particular,
$\delta_{c_i+c_j}$ can be nonzero only when both $c_i$ and $c_j$
vanish. The allowed values of the resonance count are consequently
\begin{equation}
m\in\{0,1,2,3\}.
\end{equation}
They are referred to in the Table \ref{tab:ep_cases}
\begin{table}[ht]
\centering
\caption{Different cases of the Cartan parameters and the corresponding
average entangling power $\overline{e_p}$. The values $m=4$ and $m=5$ are excluded
by the Weyl-chamber constraints.}
\label{tab:ep_cases}
\begin{tabular}{c c c}
\hline
$m$ & Conditions on $(c_1,c_2,c_3)$ & $\overline{e_p}$ \\
\hline
$0$ & $c_1 \neq c_2 \neq c_3$ & $1/2$ \\[6pt]
$1$ & $c_3 < c_2 < c_1$ & $5/12$ \\[6pt]
$2$ & $c_1>0,\quad c_2=c_3=0$ & $1/3$ \\[6pt]
$3$ & $c_1=c_2=c_3>0$ & $1/4$ \\[6pt]

\hline
\end{tabular}
\end{table}
It follows that the infinite-time averaged entangling power is
restricted to the discrete set
\begin{equation}
\lim_{n\rightarrow\infty}\overline{e_p}(U^n)
\in
\left\{
\frac{1}{2},
\frac{5}{12},
\frac{1}{3},
\frac{1}{4}
\right\},
\label{eq:ep_discrete}
\end{equation}
corresponding to $m=0,1,2,3$, respectively. Thus, the infinite-time average cannot assume an arbitrary value within the
Weyl chamber. The case $m=6$ corresponds to the identity gate, with $c_1=c_2=c_3=0$. For the SWAP gate, $c_1=c_2=c_3=\pi/4$, which has zero entangling power and therefore does not generate nontrivial entanglement dynamics under repeated time evolution.  In particular, equality with the two-qubit Lubkin value,
$\langle  \mathcal{E}\rangle=2/5$, would require
\begin{equation}
\frac{1}{2}-\frac{m}{12}=\frac{2}{5},
\qquad\Longrightarrow\qquad
m=\frac{6}{5},
\end{equation}
which is incompatible with the allowed integer resonance counts.
Hence,
\begin{equation}
\lim_{n \rightarrow  \infty}\overline{e_p}(U^n)
\neq
\langle  \mathcal{E}\rangle
=
\frac{2}{5}
\end{equation}
for every point in the Weyl chamber. The closest allowed value is
$5/12\simeq0.4167$, obtained when a single resonance condition is
satisfied, for example $c_1=c_2>c_3>0$.

From this observation, it is evident that the two-qubit system does not exhibit ergodic behavior. In the limit $n\rightarrow\infty$, the time-averaged entangling power $\overline{e_p}(U^n)$ is restricted to the discrete values $0$, $0.25,0.333
$, $0.5$, and $0.4167$, whereas the Lubkin value for the average linear entropy is $\langle  \mathcal{E}\rangle=0.4$. Since none of the asymptotic values of $\overline{e_p}(U^n)$ coincides with $\langle  \mathcal{E}\rangle$, the long-time averaged entanglement does not approach the corresponding Haar-random value, indicating the absence of ergodicity in the two-qubit system.Coming to the discussion of mixing, from Eq.~\ref{eq:ep_Un}, it is evident that, in the limit $n\rightarrow\infty$, the entangling power $e_p(U^n)$ does not saturate to a stationary value. This is because $e_p(U^n)$ is composed of undamped cosine terms, which continue to oscillate under repeated applications of the unitary. Consequently, the entangling power does not approach a well-defined asymptotic value at long times, indicating that the two-qubit dynamics is not mixing. Furthermore, since mixing is a stronger property than ergodicity, the absence of ergodicity established above directly implies the absence of mixing.

\section{Many-body System: Kicked Ising Model}
\label{sec_3}
We are extending our study to a many-body system. Specifically, we consider a kicked Ising spin chain, consisting of spin-$1/2$ particles arranged on a one-dimensional lattice of $L$ sites ($L$ is taken to be even integer) and governed by a periodically kicked Ising Hamiltonian. The corresponding time-dependent Hamiltonian is given by
\begin{equation}
H(t)
=
H_{0}
+
V
\sum_{k=-\infty}^{\infty}
\delta\left(k-\frac{t}{\tau}\right),
\end{equation}
where the static part of the Hamiltonian is
\begin{equation}
H_{0}
=
\sum_{j=1}^{L-1}
\sigma_{j}^{z}\sigma_{j+1}^{z}
+
\sum_{j=1}^{L}
h_{j}^{z}\sigma_{j}^{z},
\end{equation}
and the periodically applied kick is
\begin{equation}
V
=
\sum_{j=1}^{L}
\left(
h_{j}^{x}\sigma_{j}^{x}
+
h_{j}^{y}\sigma_{j}^{y}
\right).
\end{equation}
Here, $\sigma_{j}^{\alpha}$, with $\alpha=x,y,z$, denotes the Pauli operator acting on the $j$th site, and $\tau$ represents the time interval between two successive kicks. We can express the state of the system on the $ n+1$th kick after the $n$th kick by the unitary Floquet operator
\begin{equation}
|\psi(n+1)\rangle
=
U(\tau,\mathbf{h})
|\psi(n)\rangle,
\end{equation}
where the Floquet operator is given by, 
\begin{equation}
U(\tau,\mathbf{h})
=
e^{-i\tau V}
e^{-i\tau H_{0}}.
\end{equation}
After $n$ kicks, the time-evolved state is therefore
\begin{equation}
|\psi(n)\rangle
=
U^{n}(\tau,\mathbf{h})
|\psi(0)\rangle.
\end{equation}
The integrability of the system is determined by the choice of the external magnetic field $\mathbf{h}=(h_x,h_y,h_z)$. It is well known that  this model is integrable in the following two limits: 
1)  
$\mathbf{h}=(h_x=1,h_y=0,h_z=0)$, and 
2)
$\mathbf{h}=(h_x=0,h_y=1,h_z=1)$, and 
otherwise non-integrable, exhibits quantum-chaotic behavior~\cite{arul.2005}. To investigate the departure from integrability, we parametrize the field as $\mathbf{h}=(h_x,h_y,h_z)=\left(1-\delta,\sqrt{\delta(1-\delta/2)},\sqrt{\delta(1-\delta/2)}\right)$, such that $h_x^2+h_y^2+h_z^2=1$ for $0\leq\delta\leq1$. In our numerical analysis, we vary $\delta$ to investigate the departure from integrability. In addition, we focus on two specific parameter sets, which we refer to as Set I and Set NI. Set I corresponds to $\mathbf{h}=(h_x,h_y,h_z)=(1,0,0)$ and represents an integrable point, while Set NI corresponds to $\mathbf{h}=(h_x,h_y,h_z)=( 0.9045, 0.3457,0.8090)$, a non-integrable point. We choose the latter parameter set to facilitate comparison with the results reported in Ref.~\cite{Pal2018}. Without loss of generality, most of our numerical calculations are performed with $\tau=\pi/4$, since some of the analytical results simplify in this limit.

\paragraph*{Integrable limit:}First, we focus on the integrable limit, corresponding to Set I. This limit is exactly solvable. In this case, the operator Schmidt rank of $U^n$ increases as $2^n$ for $1\leq n\leq L$, with all the associated Schmidt coefficients being identical. Consequently, the linear operator entanglement entropy can be obtained in closed form~\cite{Pal2018}:
\begin{equation}
E(U^{n})
=1-2^{-n},
\end{equation}
Similarly, the entangling power of the $n$-th power of the unitary is given by~\cite{Pal2018}
\begin{equation}
e_p(U^n)=
\frac{1+2^L-2^{L-n}-2^{n-1}}
{\left(1+2^{L/2}\right)^2}.
\label{eq:eqn}
\end{equation}

Moreover, since, in this limit, the entangling power is periodic in $n$ with period $2L$, the long-time average of $e_p$ can be obtained by averaging over a single complete period, i.e.,

\begin{equation}
\lim_{n\rightarrow\infty}\overline{e_p}(U^n)
=
\frac{1}{2L}
\sum_{k=0}^{2L-1} e_p(U^k)=\overline{e_p}^{sat}.
\end{equation}

Furthermore, the entangling power satisfies the reflection symmetry
$e_p(U^k)=e_p(U^{2L-k})$,  and $e_p(U^0)=0$. The long-time average can therefore be written as
\begin{equation}
\overline{e_p}^{\,\mathrm{sat}}
=
\frac{1}{2L}
\left[
2\sum_{k=1}^{L-1}e_p(U^k)
+e_p(U^L)
\right].
\label{eq:periodic_average}
\end{equation}

Substituting Eq.~(\ref{eq:eqn}) into Eq.~(\ref{eq:periodic_average}), we obtain
\begin{equation}
\begin{split}
\overline{e_p}^{\,\mathrm{sat}}
=
\frac{1}{
2L\left(1+2^{L/2}\right)^2
}
\Bigg[
&2\sum_{k=1}^{L-1}
\left(
1+2^L-2^{L-k}-2^{k-1}
\right)
\\
&+2^{L-1}
\Bigg].
\end{split}
\label{eq:sat_sum}
\end{equation}
Evaluating the geometric sums gives
\begin{equation}
\begin{split}
\overline{e_p}^{\,\mathrm{sat}}
=
\frac{
4L+8+(4L-9)2^L
}{
4L\left(1+2^{L/2}\right)^2
}.
\end{split}
\label{eq:sat_exact}
\end{equation}
Therefore, in the large $L$ limit,  the leading large-$L$ behavior of the saturation value is
\begin{equation}
\overline{e_p}^{\,\mathrm{sat}}
=
1-\frac{9}{4L}
+\mathcal{O}\!\left(2^{-L/2}\right).
\label{ep sat int}
\end{equation}
It also implies, in the $L\to \infty$ limit, $\overline{e_p}^{\,\mathrm{sat}} \to 1$, which also coincides with Lubkin’s result, i.e., Haar average linear entropy  $\langle {\mathcal{E}} \rangle$ (see Sec.~\ref{sec:formalism}). This implies that, within our framework, the integrable limit of the model is also ergodic in the thermodynamic limit. However, since the entangling power is periodic in $n$ with period $2L$, the limit $
\lim_{n\rightarrow\infty} e_p(U^n)$
does not exist (although $\lim_{n\rightarrow\infty}\overline{e_p}(U^n)$ still exists). This indicates that the dynamics of the entangling power does not belong to the mixing class. In short, our results for the integrable limit show that, in the thermodynamic limit $L\to\infty$,
\begin{equation}
\lim_{n\to\infty}\overline{e_p}(U^n)
=\langle\mathcal{E}\rangle=1,
\nonumber
\end{equation}
indicating ergodic dynamics. In contrast,
\begin{equation}
\lim_{n\to\infty}e_p(U^n)
\quad\text{does not exist},
\nonumber
\end{equation}
indicating non-mixing dynamics.

\paragraph*{RMT model:} While the integrable limit is exactly solvable, in the $\delta>0$ regime, the model is non-integrable, and an analytical expression for $e_p$ is no longer available as before. Hence, numerical analysis is the only viable approach. In this regime, however, a hybrid RMT model, as proposed in Ref.~\cite{Pal2018}, provides a useful framework for predicting the dynamics of the entangling power. We first introduce this hybrid RMT model and its analytical predictions. We then compare the results of our model in the non-integrable regime with those obtained from the RMT model.

The spin chain is divided into two equal subsystems,
\begin{equation}
A=\{1,\ldots,L/2\},
\qquad
B=\{L/2+1,\ldots,L\},
\end{equation}
such that the Hilbert-space dimension of each subsystem is $2^{L/2}$. In general, the Floquet operator can be expressed schematically as
\begin{equation}
U(\tau,\mathbf{h})
=
\left(U_{A}\otimes U_{B}\right)
U_{AB}(\tau),
\end{equation}
where $U_{A}$ and $U_{B}$ describe dynamics internal to the two subsystems, while the interaction across the bipartition is generated by
\begin{equation}
U_{AB}(\tau)
=
\exp\left(
-i\tau
\sigma_{L/2}^{z}
\sigma_{L/2+1}^{z}
\right).
\end{equation}

Since the $U_{AB}(\tau)$ is a non-local unitary gate, it produces an entanglement between A and B, and the entanglement will grow in each kick. The RMT
model $U_{RMT}(\tau )$ is a hybrid one wherein we replace $U_A$, $U_B$
by local random unitary matrices and retain the interaction
as-is. As $U_A$ and $U_B$ are merely $L/2$ length chains of the original
system, this can be expected to be a reasonable model if $U(\tau,\mathbf{h})$
is sufficiently non-integrable (i.e., $\delta >0$) and possesses random matrix properties.

Within this hybrid RMT model, one can derive $\overline{e_p}(U^n)$ using the expression for $e_p(U^n)$ given in Ref.~\cite{Pal2018}. For $\tau=\pi/4$, this expression reduces to
\begin{equation}
e^{\mathrm{RMT}}_p(U^n)=1-\frac{1}{2^n}=1-e^{-n\ln2}.
\label{ep mix rmt sat}
\end{equation}
in the thermodynamic limit. 
Consequently, the corresponding average is
\begin{equation}
\overline{e^{\mathrm{RMT}}_p}(U^n)
=\frac{n-1+2^{-n}}{n+1}
=\frac{1-\frac{1}{n}+\frac{2^{-n}}{n}}{1+\frac{1}{n}}.
\label{ep sat rmt}
\end{equation}
Thus, unlike in the integrable case, $\lim_{n\to\infty} e^{\mathrm{RMT}}_p(U^n)$ exists and approaches unity. The same holds for the time-averaged entangling power, for which $\lim_{n\to\infty}\overline{e^{\mathrm{RMT}}_p}(U^n)=1$. 
In short, our results for the RMT model shows that, in the thermodynamic limit $L\to\infty$,
\begin{equation}
\lim_{n\to\infty}\overline{e_p}(U^n)
=\langle\mathcal{E}\rangle=1,
\nonumber
\end{equation}
indicating ergodic dynamics. Also, 
\begin{equation}
\lim_{n\to\infty}e_p(U^n)
=\langle\mathcal{E}\rangle=1,
\nonumber
\end{equation}
indicating mixing dynamics.

This suggests that the entangling-power dynamics in the RMT model is not only ergodic but also mixing. In addition, the distance of $e^{\mathrm{RMT}}_p(U^n)$ from its fixed point at $1$ decays exponentially with $n$ (see Eq.~\ref{ep mix rmt sat}), whereas the leading finite-$n$ correction to the time-averaged entangling power is of order $1/n$:
\begin{equation}
\overline{e^{\mathrm{RMT}}_p}(U^n)=1-\frac{2}{n}+O\left(\frac{1}{n^2}\right).
\end{equation}
These RMT results provide a natural motivation for our analysis of the finite-$L$ numerical data in the non-integrable regime later, where the exact analytical prediction is not possible. In particular, at large $n$, we are going to fit the entangling power and its time average to the forms
\begin{equation}
e_p(U^n)=f-e^{-\lambda n},
\label{fitting mixing}
\end{equation}
and
\begin{equation}
\overline{e_p}(U^n)=f-\frac{B}{n}.
\label{fitting ergodic}
\end{equation}
In the thermodynamic limit, $L\to\infty$, we therefore expect the fitting parameters to approach their RMT values,
\begin{equation}
f\to 1,\qquad \lambda\to\ln 2,\qquad B\to 2,
\label{eq: thermodynamic fitting para}
\end{equation}
in the non-integrable regime. Next, we check our prediction using numerical calculations.

\begin{figure}
    \centering
    \includegraphics[width=1 \linewidth]{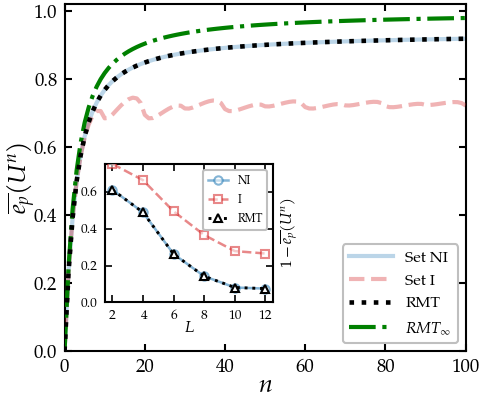}
      \caption{Entangling power of the Floquet time-evolution operator versus the number of time steps $n$, for $L=10$. 
   Results are presented for integrable set I, non-integrable Set NI,  finite and infinite-sized RMT model.
   Inset show $1-\overline{e_p}(U^n)$ vs the system size $L$,for $n=100$.}
    \label{L10 spin chain}
\end{figure}

\paragraph*{Numerical results:} Figure~\ref{L10 spin chain} shows the variation of the time-averaged entangling power, $\overline{e}_p(U^n)$, with the number of Floquet cycles $n$ for $L=10$. The first remarkable observation is that the data for the non-integrable Set NI are almost indistinguishable from those of the finite-sized RMT model. In all cases, $\overline{e}_p(U^n)$ initially grows rapidly with $n$ and then gradually approaches the fixed point at $1$. For reasonably large $n$, the integrable data consistently remain below those for the non-integrable and RMT cases, indicating a slower approach towards the fixed point. 
The inset shows the dependence of the time-averaged entangling power on the system size $L$ at a fixed, reasonably large $n=100$. With increasing $L$, the data for both the integrable and non-integrable cases become progressively closer to the fixed point at $1$. However, for every system size considered, the non-integrable results remain closer to $1$ than the corresponding integrable results. In the strict $n\to\infty$ and $L\to \infty$ limit, $\overline{e}_p(U^n)$ approaches $1$.
These numerical results indicate that the entangling power dynamics is ergodic in nature in both the integrable and non-integrable regimes.

\begin{figure}
    \centering
\includegraphics[width=1\linewidth]{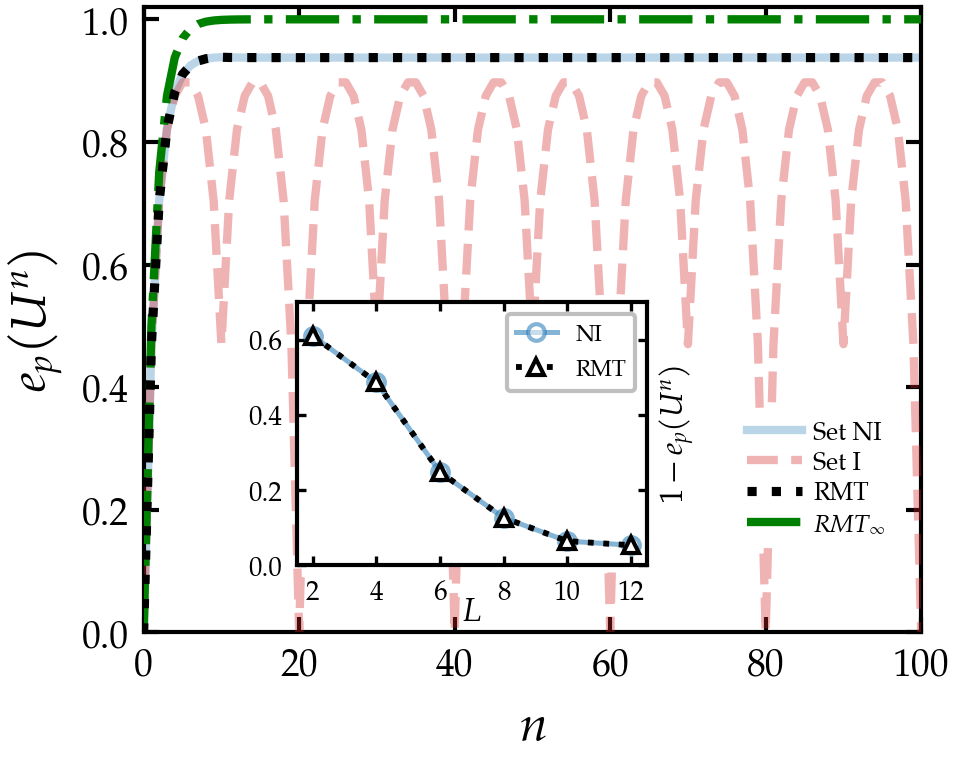}
   \caption{Entangling power of the Floquet time-evolution operator versus the number of time steps $n$, for $L=10$. 
   Results are presented for integrable set I, non-integrable Set NI,  finite and infinite-sized RMT model.
   Inset show $1-e_p(U^n)$ vs the system size $L$,for $n=100$.}
    \label{L10 spin chain normal}
\end{figure}

Next, we perform a similar analysis for ${e}_p(U^n)$ by plotting ${e}_p(U^n)$ as a function of $n$ in Fig.~\ref{L10 spin chain normal}. Once again, we find that the finite-size RMT results agree remarkably well with the Set NI non-integrable data. Moreover, the inset shows that, for $n=100$, the data move progressively closer to $1$ as $L$ increases, suggesting that, in the strict $n\to\infty$ and $L\to\infty$ limits, $e_p(U^n)$ approaches $1$. In contrast, the integrable data continue to exhibit oscillatory behavior, as expected from our previous analytical results. These observations confirm that the entangling-power dynamics in the non-integrable regime not only belongs to the ergodic class but also exhibits mixing, whereas mixing is absent in the integrable limit.

\begin{figure}
    \centering
    \includegraphics[width=1 \linewidth]{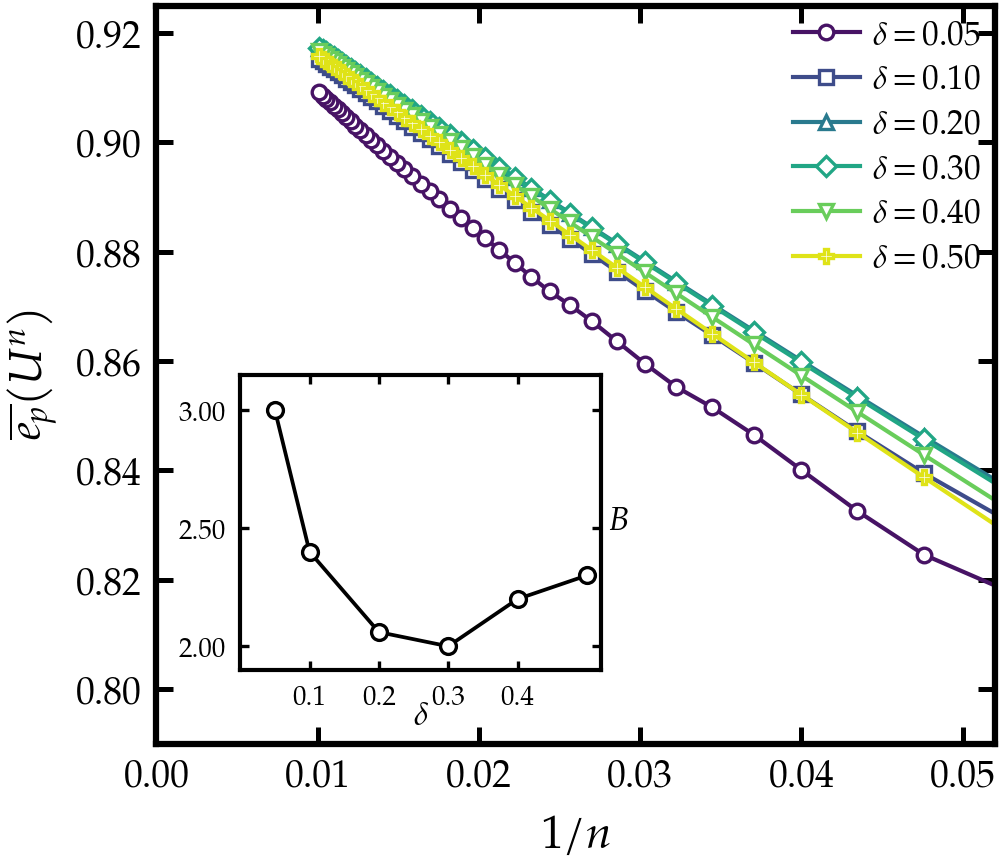}
    
    \caption{Average entangling power vs $1/n$, here $n$ ranges from $20$ to $100$ for different values of $\delta$. Data is fitted with a function $f-B/n$ (where $f$ and $B$ are fitting parameters), and inset shows  $B$ vs $\delta$ plot.}
    \label{L8}
\end{figure}

Further, we study the approach of $\overline{e}_p(U^n)$ to its fixed point at large $n$ for different values of integrability breaking parameter $\delta$ in Fig.~\ref{L8}. Motivated by the numerical data and the RMT prediction discussed before, we fit the data using 
$\overline{e}_p(U^n)=f-\frac{B}{n}$,
(using the fitting funtion proposed in Eq.~\ref{fitting ergodic}),
where $f$ denotes the asymptotic fixed-point value. Since, for a given $n$, $\overline{e}_p(U^n)$ is smaller in the integrable limit than in the non-integrable regime, we expect a larger value of $B$ in the integrable limit. Consequently, as we move away from the integrable point $\delta=0$, the value of $B$ is expected to decrease. In the thermodynamic limit, we further expect $B\to 2$ (see Eq.~\ref{eq: thermodynamic fitting para}), in agreement with the RMT prediction.
 We observe precisely this behavior in the insets of Fig.~\ref{L8}. At $\delta=0$, $B$ takes its maximum value and subsequently decreases with increasing integrability-breaking parameter $\delta$. It reaches a maximum around $\delta\approx 0.3$, where the fitted value of $B$ becomes very close to the RMT prediction, $B=2$. Moreover, the best-fit values of $f$ remain consistently close to $1$ for all values of $\delta$.
\begin{figure}
    \centering
    \includegraphics[width=1 \linewidth]{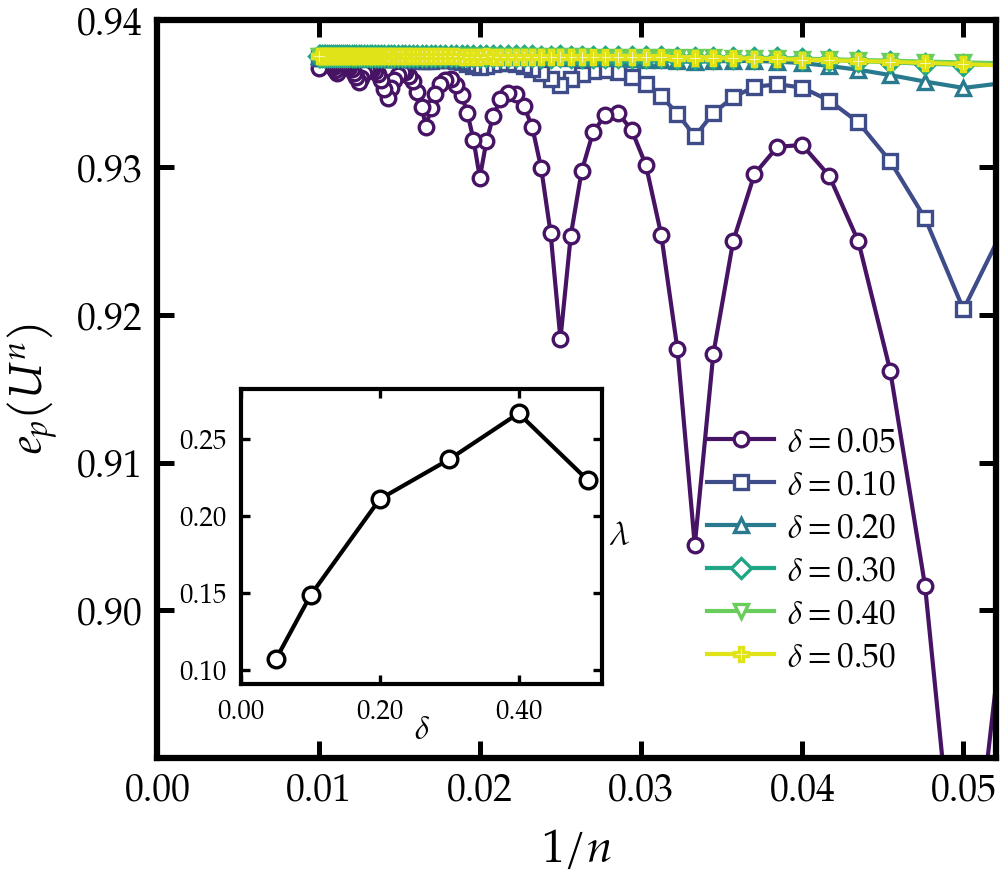}
    
    \caption{Entangling power vs $1/n$, here $n$ ranges from $20$ to $100$ for different values of $\delta$. Inset shows the variation of the fitting parameter $\lambda$ (which we refer as Lyapunov like like exponent) with $\delta$. }
    \label{ep_delta}
\end{figure}

Figure.~\ref{ep_delta} shows an analysis similar to that presented in Fig.~\ref{L8} for the entangling power ${e}_p(U^n)$. We find that, for small $n$, the data exhibits pronounced oscillations for small $\delta$, i.e., close to the integrable point. However, these fluctuations decay rapidly with increasing $n$, suggesting convergence toward the $n\to\infty$ limit of ${e}_p(U^n)$, indicating signature of mixing. In contrast, sufficiently away from the integrable point, corresponding to reasonably large $\delta$, the fluctuations are strongly suppressed even for small $n$. We perform an exponential fit of the data, as predicted by the RMT result (see Eqs.~\ref{ep mix rmt sat} and~\ref{fitting mixing}), to extract the decay rate $\lambda$. The inset shows the variation of $\lambda$ with the integrability-breaking parameter $\delta$. We find that $\lambda$ increases with increasing $\delta$, and expects that it will approaches to $\ln 2$ in the thermodynamic limit, which is the RMT prediction (see Eq.~\ref{eq: thermodynamic fitting para}). This behavior is consistent with the quantum chaotic dynamics. Moreover, as discussed in Sec.~\ref{sec:formalism}, $\lambda$ can be identified with the Lyapunov like like exponent characterizing the underlying chaotic dynamics.

\section{Conclusion}

In this work, we have characterized ergodicity and mixing in quantum dynamics from a quantum-information perspective, using entanglement generation as the principal diagnostic. In particular, ergodicity is examined through the long-time-averaged entangling power and its comparison with Lubkin's Haar-averaged linear entropy, whereas mixing is characterized by the long-time behavior of the entangling power. This approach provides an entanglement-based framework for distinguishing between the two dynamical properties and, in particular, allows systems that are ergodic but non-mixing to be identified.

We first applied this framework to general two-qubit gates, using their canonical representation in the Weyl chamber. The infinite-time-averaged entangling power is determined by the resonance conditions among the Weyl-chamber parameters and is therefore restricted to a discrete set of values. None of these values coincides with the corresponding two-qubit Lubkin value. Thus, within the present entanglement-based criterion, two-qubit gates do not exhibit ergodicity. Furthermore, the entangling power under repeated applications of a fixed two-qubit gate contains persistent oscillatory contributions and, in general, does not relax to the Haar-averaged value in the long-time limit. The two-qubit dynamics is therefore also non-mixing.

We further extended the analysis to the kicked Ising spin chain. For the particular integrable point considered here, the long-time-averaged entangling power approaches the Lubkin value as the chain length increases, indicating ergodic behavior in the thermodynamic limit according to our entanglement-based criterion. The underlying dynamics remain periodic, with a recurrence period proportional to the system size. The integrable kicked Ising chain, therefore, provides an explicit example of dynamics that becomes ergodic but remains non-mixing in the large-system limit. The non-integrable model is consistent with the RMT predictions, and it remains mixing and ergodic in the high-dimensional system. We also study the chaotic dynamics by calculating the Lyapunov-like exponent, which increases as the system becomes less integrable.

We can extend this work in the case of multiple bipartitions. Since we are going to a higher number of qubits, the natural number of bipartitions will increase, and the unitary evolution can create entanglement in the same state in different ways across different bipartitions. In the two-qubit gates, we found that there are four fixed points, and we can ask what possible classes of long-time-averaged entangling power can be attained by  multiqubit systems.

An interesting question is to investigate and relate the role of ergodic and mixing behavior of entangling power in quantum circuits comprising local gates for the generation of unitary $t$-designs \cite{Gross2007UnitaryDesigns,9tnv-2559} and as well as for other solvable models of many-body systems \cite{rampp2025solvablequantumcircuitsspacetime,rampp2026infinitelevelhierarchysolvablequantum,yu2024hierarchical}.

\begin{acknowledgments}
R.M. acknowledges the DST-Inspire fellowship from the
Department of Science and Technology, Government of India. We acknowledge Jisha C. for useful discussions on the random matrix theory (RMT) formalism. SA acknowledge funding support from the National Quantum Mission, an initiative of the Department of Science and Technology, Govt. of India. We also acknowledge the support provided by the Foundation for QC Innovation (FQCI), DST-NQM T-Hub at IISc Bengaluru, in facilitating this project.

\end{acknowledgments}

\newpage

\begin{thebibliography}{58}%
\makeatletter
\providecommand \@ifxundefined [1]{%
 \@ifx{#1\undefined}
}%
\providecommand \@ifnum [1]{%
 \ifnum #1\expandafter \@firstoftwo
 \else \expandafter \@secondoftwo
 \fi
}%
\providecommand \@ifx [1]{%
 \ifx #1\expandafter \@firstoftwo
 \else \expandafter \@secondoftwo
 \fi
}%
\providecommand \natexlab [1]{#1}%
\providecommand \enquote  [1]{``#1''}%
\providecommand \bibnamefont  [1]{#1}%
\providecommand \bibfnamefont [1]{#1}%
\providecommand \citenamefont [1]{#1}%
\providecommand \href@noop [0]{\@secondoftwo}%
\providecommand \href [0]{\begingroup \@sanitize@url \@href}%
\providecommand \@href[1]{\@@startlink{#1}\@@href}%
\providecommand \@@href[1]{\endgroup#1\@@endlink}%
\providecommand \@sanitize@url [0]{\catcode `\\12\catcode `\$12\catcode
  `\&12\catcode `\#12\catcode `\^12\catcode `\_12\catcode `\%12\relax}%
\providecommand \@@startlink[1]{}%
\providecommand \@@endlink[0]{}%
\providecommand \url  [0]{\begingroup\@sanitize@url \@url }%
\providecommand \@url [1]{\endgroup\@href {#1}{\urlprefix }}%
\providecommand \urlprefix  [0]{URL }%
\providecommand \Eprint [0]{\href }%
\providecommand \doibase [0]{http://dx.doi.org/}%
\providecommand \selectlanguage [0]{\@gobble}%
\providecommand \bibinfo  [0]{\@secondoftwo}%
\providecommand \bibfield  [0]{\@secondoftwo}%
\providecommand \translation [1]{[#1]}%
\providecommand \BibitemOpen [0]{}%
\providecommand \bibitemStop [0]{}%
\providecommand \bibitemNoStop [0]{.\EOS\space}%
\providecommand \EOS [0]{\spacefactor3000\relax}%
\providecommand \BibitemShut  [1]{\csname bibitem#1\endcsname}%
\let\auto@bib@innerbib\@empty
\bibitem [{\citenamefont {Arnold}\ and\ \citenamefont
  {Avez}(1968)}]{ArnoldAvez1968}%
  \BibitemOpen
  \bibfield  {author} {\bibinfo {author} {\bibfnamefont {V.~I.}\ \bibnamefont
  {Arnold}}\ and\ \bibinfo {author} {\bibfnamefont {A.}~\bibnamefont {Avez}},\
  }\href@noop {} {\emph {\bibinfo {title} {Ergodic Problems of Classical
  Mechanics}}}\ (\bibinfo  {publisher} {W. A. Benjamin},\ \bibinfo {address}
  {New York},\ \bibinfo {year} {1968})\BibitemShut {NoStop}%
\bibitem [{\citenamefont {Gisin}\ \emph {et~al.}(2002)\citenamefont {Gisin},
  \citenamefont {Ribordy}, \citenamefont {Tittel},\ and\ \citenamefont
  {Zbinden}}]{Gisin2007}%
  \BibitemOpen
  \bibfield  {author} {\bibinfo {author} {\bibfnamefont {Nicolas}\ \bibnamefont
  {Gisin}}, \bibinfo {author} {\bibfnamefont {Gr{\'e}goire}\ \bibnamefont
  {Ribordy}}, \bibinfo {author} {\bibfnamefont {Wolfgang}\ \bibnamefont
  {Tittel}}, \ and\ \bibinfo {author} {\bibfnamefont {Hugo}\ \bibnamefont
  {Zbinden}},\ }\bibfield  {title} {\enquote {\bibinfo {title} {Quantum
  cryptography},}\ }\href {\doibase 10.1103/RevModPhys.74.145} {\bibfield
  {journal} {\bibinfo  {journal} {Reviews of Modern Physics}\ }\textbf
  {\bibinfo {volume} {74}},\ \bibinfo {pages} {145--195} (\bibinfo {year}
  {2002})}\BibitemShut {NoStop}%
\bibitem [{\citenamefont {Nielsen}\ and\ \citenamefont
  {Chuang}(2010)}]{Nielsen2010}%
  \BibitemOpen
  \bibfield  {author} {\bibinfo {author} {\bibfnamefont {Michael~A}\
  \bibnamefont {Nielsen}}\ and\ \bibinfo {author} {\bibfnamefont {Isaac~L}\
  \bibnamefont {Chuang}},\ }\href@noop {} {\emph {\bibinfo {title} {Quantum
  computation and quantum information}}}\ (\bibinfo  {publisher} {Cambridge
  university press},\ \bibinfo {year} {2010})\BibitemShut {NoStop}%
\bibitem [{\citenamefont {Amico}\ \emph {et~al.}(2008)\citenamefont {Amico},
  \citenamefont {Fazio}, \citenamefont {Osterloh},\ and\ \citenamefont
  {Vedral}}]{Amico2008}%
  \BibitemOpen
  \bibfield  {author} {\bibinfo {author} {\bibfnamefont {L.}~\bibnamefont
  {Amico}}, \bibinfo {author} {\bibfnamefont {R.}~\bibnamefont {Fazio}},
  \bibinfo {author} {\bibfnamefont {A.}~\bibnamefont {Osterloh}}, \ and\
  \bibinfo {author} {\bibfnamefont {V.}~\bibnamefont {Vedral}},\ }\bibfield
  {title} {\enquote {\bibinfo {title} {Entanglement in many-body systems},}\
  }\href {\doibase 10.1103/RevModPhys.80.517} {\bibfield  {journal} {\bibinfo
  {journal} {Rev. Mod. Phys.}\ }\textbf {\bibinfo {volume} {80}},\ \bibinfo
  {pages} {517--576} (\bibinfo {year} {2008})}\BibitemShut {NoStop}%
\bibitem [{\citenamefont {Osterloh}\ \emph {et~al.}(2002)\citenamefont
  {Osterloh}, \citenamefont {Amico}, \citenamefont {Falci},\ and\ \citenamefont
  {Fazio}}]{Osterloh2002}%
  \BibitemOpen
  \bibfield  {author} {\bibinfo {author} {\bibfnamefont {A.}~\bibnamefont
  {Osterloh}}, \bibinfo {author} {\bibfnamefont {L.}~\bibnamefont {Amico}},
  \bibinfo {author} {\bibfnamefont {G.}~\bibnamefont {Falci}}, \ and\ \bibinfo
  {author} {\bibfnamefont {R.}~\bibnamefont {Fazio}},\ }\bibfield  {title}
  {\enquote {\bibinfo {title} {Scaling of entanglement close to a quantum phase
  transition},}\ }\href {\doibase 10.1038/416608a} {\bibfield  {journal}
  {\bibinfo  {journal} {Nature}\ }\textbf {\bibinfo {volume} {416}},\ \bibinfo
  {pages} {608--610} (\bibinfo {year} {2002})}\BibitemShut {NoStop}%
\bibitem [{\citenamefont {Osborne}\ and\ \citenamefont
  {Nielsen}(2002)}]{Osborne2002}%
  \BibitemOpen
  \bibfield  {author} {\bibinfo {author} {\bibfnamefont {T.~J.}\ \bibnamefont
  {Osborne}}\ and\ \bibinfo {author} {\bibfnamefont {M.~A.}\ \bibnamefont
  {Nielsen}},\ }\bibfield  {title} {\enquote {\bibinfo {title} {Entanglement in
  a simple quantum phase transition},}\ }\href {\doibase
  10.1103/PhysRevA.66.032110} {\bibfield  {journal} {\bibinfo  {journal} {Phys.
  Rev. A}\ }\textbf {\bibinfo {volume} {66}},\ \bibinfo {pages} {032110}
  (\bibinfo {year} {2002})}\BibitemShut {NoStop}%
\bibitem [{\citenamefont {Deutsch}(1991)}]{Deutsch1991}%
  \BibitemOpen
  \bibfield  {author} {\bibinfo {author} {\bibfnamefont {J.~M.}\ \bibnamefont
  {Deutsch}},\ }\bibfield  {title} {\enquote {\bibinfo {title} {Quantum
  statistical mechanics in a closed system},}\ }\href {\doibase
  10.1103/PhysRevA.43.2046} {\bibfield  {journal} {\bibinfo  {journal} {Phys.
  Rev. A}\ }\textbf {\bibinfo {volume} {43}},\ \bibinfo {pages} {2046--2049}
  (\bibinfo {year} {1991})}\BibitemShut {NoStop}%
\bibitem [{\citenamefont {Srednicki}(1994)}]{Srednicki1994}%
  \BibitemOpen
  \bibfield  {author} {\bibinfo {author} {\bibfnamefont {M.}~\bibnamefont
  {Srednicki}},\ }\bibfield  {title} {\enquote {\bibinfo {title} {Chaos and
  quantum thermalization},}\ }\href {\doibase 10.1103/PhysRevE.50.888}
  {\bibfield  {journal} {\bibinfo  {journal} {Phys. Rev. E}\ }\textbf {\bibinfo
  {volume} {50}},\ \bibinfo {pages} {888--901} (\bibinfo {year}
  {1994})}\BibitemShut {NoStop}%
\bibitem [{\citenamefont {D'Alessio}\ \emph {et~al.}(2016)\citenamefont
  {D'Alessio}, \citenamefont {Kafri}, \citenamefont {Polkovnikov},\ and\
  \citenamefont {Rigol}}]{DAlessio2016}%
  \BibitemOpen
  \bibfield  {author} {\bibinfo {author} {\bibfnamefont {L.}~\bibnamefont
  {D'Alessio}}, \bibinfo {author} {\bibfnamefont {Y.}~\bibnamefont {Kafri}},
  \bibinfo {author} {\bibfnamefont {A.}~\bibnamefont {Polkovnikov}}, \ and\
  \bibinfo {author} {\bibfnamefont {M.}~\bibnamefont {Rigol}},\ }\bibfield
  {title} {\enquote {\bibinfo {title} {From quantum chaos and eigenstate
  thermalization to statistical mechanics and thermodynamics},}\ }\href
  {\doibase 10.1080/00018732.2016.1198134} {\bibfield  {journal} {\bibinfo
  {journal} {Advances in Physics}\ }\textbf {\bibinfo {volume} {65}},\ \bibinfo
  {pages} {239--362} (\bibinfo {year} {2016})}\BibitemShut {NoStop}%
\bibitem [{\citenamefont {Hayden}\ and\ \citenamefont
  {Preskill}(2007)}]{Hayden2007}%
  \BibitemOpen
  \bibfield  {author} {\bibinfo {author} {\bibfnamefont {P.}~\bibnamefont
  {Hayden}}\ and\ \bibinfo {author} {\bibfnamefont {J.}~\bibnamefont
  {Preskill}},\ }\bibfield  {title} {\enquote {\bibinfo {title} {Black holes as
  mirrors: quantum information in random subsystems},}\ }\href {\doibase
  10.1088/1126-6708/2007/09/120} {\bibfield  {journal} {\bibinfo  {journal}
  {JHEP}\ }\textbf {\bibinfo {volume} {09}},\ \bibinfo {pages} {120} (\bibinfo
  {year} {2007})}\BibitemShut {NoStop}%
\bibitem [{\citenamefont {Sekino}\ and\ \citenamefont
  {Susskind}(2008)}]{Sekino2008}%
  \BibitemOpen
  \bibfield  {author} {\bibinfo {author} {\bibfnamefont {Y.}~\bibnamefont
  {Sekino}}\ and\ \bibinfo {author} {\bibfnamefont {L.}~\bibnamefont
  {Susskind}},\ }\bibfield  {title} {\enquote {\bibinfo {title} {Fast
  scramblers},}\ }\href {\doibase 10.1088/1126-6708/2008/10/065} {\bibfield
  {journal} {\bibinfo  {journal} {JHEP}\ }\textbf {\bibinfo {volume} {10}},\
  \bibinfo {pages} {065} (\bibinfo {year} {2008})}\BibitemShut {NoStop}%
\bibitem [{\citenamefont {Swingle}(2018)}]{Swingle2018}%
  \BibitemOpen
  \bibfield  {author} {\bibinfo {author} {\bibfnamefont {B.}~\bibnamefont
  {Swingle}},\ }\bibfield  {title} {\enquote {\bibinfo {title} {Unscrambling
  the physics of out-of-time-order correlators},}\ }\href {\doibase
  10.1038/s41567-018-0295-5} {\bibfield  {journal} {\bibinfo  {journal} {Nature
  Physics}\ }\textbf {\bibinfo {volume} {14}},\ \bibinfo {pages} {988--990}
  (\bibinfo {year} {2018})}\BibitemShut {NoStop}%
\bibitem [{\citenamefont {Calabrese}\ and\ \citenamefont
  {Cardy}(2005)}]{Calabrese2005}%
  \BibitemOpen
  \bibfield  {author} {\bibinfo {author} {\bibfnamefont {P.}~\bibnamefont
  {Calabrese}}\ and\ \bibinfo {author} {\bibfnamefont {J.}~\bibnamefont
  {Cardy}},\ }\bibfield  {title} {\enquote {\bibinfo {title} {Evolution of
  entanglement entropy in one-dimensional systems},}\ }\href {\doibase
  10.1088/1742-5468/2005/04/P04010} {\bibfield  {journal} {\bibinfo  {journal}
  {J. Stat. Mech.}\ ,\ \bibinfo {pages} {P04010}} (\bibinfo {year}
  {2005})}\BibitemShut {NoStop}%
\bibitem [{\citenamefont {Calabrese}\ and\ \citenamefont
  {Cardy}(2007)}]{Calabrese2007}%
  \BibitemOpen
  \bibfield  {author} {\bibinfo {author} {\bibfnamefont {P.}~\bibnamefont
  {Calabrese}}\ and\ \bibinfo {author} {\bibfnamefont {J.}~\bibnamefont
  {Cardy}},\ }\bibfield  {title} {\enquote {\bibinfo {title} {Quantum quenches
  in extended systems},}\ }\href {\doibase 10.1088/1742-5468/2007/06/P06008}
  {\bibfield  {journal} {\bibinfo  {journal} {J. Stat. Mech.}\ ,\ \bibinfo
  {pages} {P06008}} (\bibinfo {year} {2007})}\BibitemShut {NoStop}%
\bibitem [{\citenamefont {Eisert}\ \emph {et~al.}(2015)\citenamefont {Eisert},
  \citenamefont {Friesdorf},\ and\ \citenamefont {Gogolin}}]{Eisert2015}%
  \BibitemOpen
  \bibfield  {author} {\bibinfo {author} {\bibfnamefont {J.}~\bibnamefont
  {Eisert}}, \bibinfo {author} {\bibfnamefont {M.}~\bibnamefont {Friesdorf}}, \
  and\ \bibinfo {author} {\bibfnamefont {C.}~\bibnamefont {Gogolin}},\
  }\bibfield  {title} {\enquote {\bibinfo {title} {Quantum many-body systems
  out of equilibrium},}\ }\href {\doibase 10.1038/nphys3215} {\bibfield
  {journal} {\bibinfo  {journal} {Nature Physics}\ }\textbf {\bibinfo {volume}
  {11}},\ \bibinfo {pages} {124--130} (\bibinfo {year} {2015})}\BibitemShut
  {NoStop}%
\bibitem [{\citenamefont {Zanardi}\ \emph
  {et~al.}(2000{\natexlab{a}})\citenamefont {Zanardi}, \citenamefont {Zalka},\
  and\ \citenamefont {Faoro}}]{zanardi2000entangling}%
  \BibitemOpen
  \bibfield  {author} {\bibinfo {author} {\bibfnamefont {Paolo}\ \bibnamefont
  {Zanardi}}, \bibinfo {author} {\bibfnamefont {Christof}\ \bibnamefont
  {Zalka}}, \ and\ \bibinfo {author} {\bibfnamefont {Lara}\ \bibnamefont
  {Faoro}},\ }\bibfield  {title} {\enquote {\bibinfo {title} {Entangling power
  of quantum evolutions},}\ }\href@noop {} {\bibfield  {journal} {\bibinfo
  {journal} {Physical Review A}\ }\textbf {\bibinfo {volume} {62}},\ \bibinfo
  {pages} {030301} (\bibinfo {year} {2000}{\natexlab{a}})}\BibitemShut
  {NoStop}%
\bibitem [{\citenamefont
  {Zanardi}(2001{\natexlab{a}})}]{zanardi2001entanglement}%
  \BibitemOpen
  \bibfield  {author} {\bibinfo {author} {\bibfnamefont {Paolo}\ \bibnamefont
  {Zanardi}},\ }\bibfield  {title} {\enquote {\bibinfo {title} {Entanglement of
  quantum evolutions},}\ }\href@noop {} {\bibfield  {journal} {\bibinfo
  {journal} {Physical Review A}\ }\textbf {\bibinfo {volume} {63}},\ \bibinfo
  {pages} {040304} (\bibinfo {year} {2001}{\natexlab{a}})}\BibitemShut
  {NoStop}%
\bibitem [{\citenamefont {Pal}\ and\ \citenamefont
  {Lakshminarayan}(2018)}]{Pal2018}%
  \BibitemOpen
  \bibfield  {author} {\bibinfo {author} {\bibfnamefont {Rajarshi}\
  \bibnamefont {Pal}}\ and\ \bibinfo {author} {\bibfnamefont {Arul}\
  \bibnamefont {Lakshminarayan}},\ }\bibfield  {title} {\enquote {\bibinfo
  {title} {Entangling power of time-evolution operators in integrable and
  nonintegrable many-body systems},}\ }\href {\doibase
  10.1103/PhysRevB.98.174304} {\bibfield  {journal} {\bibinfo  {journal} {Phys.
  Rev. B}\ }\textbf {\bibinfo {volume} {98}},\ \bibinfo {pages} {174304}
  (\bibinfo {year} {2018})}\BibitemShut {NoStop}%
\bibitem [{\citenamefont {Cornfeld}\ \emph {et~al.}(1982)\citenamefont
  {Cornfeld}, \citenamefont {Fomin},\ and\ \citenamefont
  {Sinai}}]{Cornfeld1982}%
  \BibitemOpen
  \bibfield  {author} {\bibinfo {author} {\bibfnamefont {I.~P.}\ \bibnamefont
  {Cornfeld}}, \bibinfo {author} {\bibfnamefont {S.~V.}\ \bibnamefont {Fomin}},
  \ and\ \bibinfo {author} {\bibfnamefont {Ya.~G.}\ \bibnamefont {Sinai}},\
  }\href@noop {} {\emph {\bibinfo {title} {Ergodic Theory}}},\ \bibinfo
  {series} {A Series of Comprehensive Studies in Mathematics}, Vol.\ \bibinfo
  {volume} {245}\ (\bibinfo  {publisher} {Springer},\ \bibinfo {address} {New
  York},\ \bibinfo {year} {1982})\BibitemShut {NoStop}%
\bibitem [{\citenamefont {Zaslavsky}(1981)}]{Zaslavsky1981}%
  \BibitemOpen
  \bibfield  {author} {\bibinfo {author} {\bibfnamefont {G.~M.}\ \bibnamefont
  {Zaslavsky}},\ }\bibfield  {title} {\enquote {\bibinfo {title} {Stochasticity
  in quantum systems},}\ }\href@noop {} {\bibfield  {journal} {\bibinfo
  {journal} {Physics Reports}\ }\textbf {\bibinfo {volume} {80}},\ \bibinfo
  {pages} {157} (\bibinfo {year} {1981})}\BibitemShut {NoStop}%
\bibitem [{\citenamefont {Peres}(1984)}]{Peres1984}%
  \BibitemOpen
  \bibfield  {author} {\bibinfo {author} {\bibfnamefont {A.}~\bibnamefont
  {Peres}},\ }\bibfield  {title} {\enquote {\bibinfo {title} {Ergodicity and
  mixing in quantum theory. i},}\ }\href@noop {} {\bibfield  {journal}
  {\bibinfo  {journal} {Physical Review A}\ }\textbf {\bibinfo {volume} {30}},\
  \bibinfo {pages} {504} (\bibinfo {year} {1984})}\BibitemShut {NoStop}%
\bibitem [{\citenamefont {Bertini}\ \emph {et~al.}(2019)\citenamefont
  {Bertini}, \citenamefont {Kos},\ and\ \citenamefont {Prosen}}]{Bertini2019}%
  \BibitemOpen
  \bibfield  {author} {\bibinfo {author} {\bibfnamefont {Bruno}\ \bibnamefont
  {Bertini}}, \bibinfo {author} {\bibfnamefont {Pavel}\ \bibnamefont {Kos}}, \
  and\ \bibinfo {author} {\bibfnamefont {Toma{\v{z}}}\ \bibnamefont {Prosen}},\
  }\bibfield  {title} {\enquote {\bibinfo {title} {Exact correlation functions
  for dual-unitary lattice models in $1+1$ dimensions},}\ }\href {\doibase
  10.1103/PhysRevLett.123.210601} {\bibfield  {journal} {\bibinfo  {journal}
  {Phys. Rev. Lett.}\ }\textbf {\bibinfo {volume} {123}},\ \bibinfo {pages}
  {210601} (\bibinfo {year} {2019})}\BibitemShut {NoStop}%
\bibitem [{\citenamefont {Aravinda}\ \emph {et~al.}(2021)\citenamefont
  {Aravinda}, \citenamefont {Rather},\ and\ \citenamefont
  {Lakshminarayan}}]{Aravinda2021}%
  \BibitemOpen
  \bibfield  {author} {\bibinfo {author} {\bibfnamefont {S.}~\bibnamefont
  {Aravinda}}, \bibinfo {author} {\bibfnamefont {Suhail~Ahmad}\ \bibnamefont
  {Rather}}, \ and\ \bibinfo {author} {\bibfnamefont {Arul}\ \bibnamefont
  {Lakshminarayan}},\ }\bibfield  {title} {\enquote {\bibinfo {title} {From
  dual-unitary to quantum bernoulli circuits: Role of the entangling power in
  constructing a quantum ergodic hierarchy},}\ }\href {\doibase
  10.1103/PhysRevResearch.3.043034} {\bibfield  {journal} {\bibinfo  {journal}
  {Physical Review Research}\ }\textbf {\bibinfo {volume} {3}},\ \bibinfo
  {pages} {043034} (\bibinfo {year} {2021})}\BibitemShut {NoStop}%
\bibitem [{\citenamefont {Rather}\ \emph {et~al.}(2020)\citenamefont {Rather},
  \citenamefont {Aravinda},\ and\ \citenamefont {Lakshminarayan}}]{Rather2020}%
  \BibitemOpen
  \bibfield  {author} {\bibinfo {author} {\bibfnamefont {Suhail~Ahmad}\
  \bibnamefont {Rather}}, \bibinfo {author} {\bibfnamefont {S.}~\bibnamefont
  {Aravinda}}, \ and\ \bibinfo {author} {\bibfnamefont {Arul}\ \bibnamefont
  {Lakshminarayan}},\ }\bibfield  {title} {\enquote {\bibinfo {title} {Creating
  ensembles of dual unitary and maximally entangling quantum evolutions},}\
  }\href {\doibase 10.1103/PhysRevLett.125.070501} {\bibfield  {journal}
  {\bibinfo  {journal} {Physical Review Letters}\ }\textbf {\bibinfo {volume}
  {125}},\ \bibinfo {pages} {070501} (\bibinfo {year} {2020})}\BibitemShut
  {NoStop}%
\bibitem [{\citenamefont {Jonay}\ \emph {et~al.}(2021)\citenamefont {Jonay},
  \citenamefont {Khemani},\ and\ \citenamefont {Ippoliti}}]{Jonay2021}%
  \BibitemOpen
  \bibfield  {author} {\bibinfo {author} {\bibfnamefont {Cheryne}\ \bibnamefont
  {Jonay}}, \bibinfo {author} {\bibfnamefont {Vedika}\ \bibnamefont {Khemani}},
  \ and\ \bibinfo {author} {\bibfnamefont {Matteo}\ \bibnamefont {Ippoliti}},\
  }\bibfield  {title} {\enquote {\bibinfo {title} {Triunitary quantum
  circuits},}\ }\href {\doibase 10.1103/PhysRevResearch.3.043046} {\bibfield
  {journal} {\bibinfo  {journal} {Physical Review Research}\ }\textbf {\bibinfo
  {volume} {3}},\ \bibinfo {pages} {043046} (\bibinfo {year}
  {2021})}\BibitemShut {NoStop}%
\bibitem [{\citenamefont {Yu}\ \emph {et~al.}(2024)\citenamefont {Yu},
  \citenamefont {Wang},\ and\ \citenamefont {Kos}}]{yu2024hierarchical}%
  \BibitemOpen
  \bibfield  {author} {\bibinfo {author} {\bibfnamefont {Xie-Hang}\
  \bibnamefont {Yu}}, \bibinfo {author} {\bibfnamefont {Zhiyuan}\ \bibnamefont
  {Wang}}, \ and\ \bibinfo {author} {\bibfnamefont {Pavel}\ \bibnamefont
  {Kos}},\ }\bibfield  {title} {\enquote {\bibinfo {title} {Hierarchical
  generalization of dual unitarity},}\ }\href@noop {} {\bibfield  {journal}
  {\bibinfo  {journal} {Quantum}\ }\textbf {\bibinfo {volume} {8}},\ \bibinfo
  {pages} {1260} (\bibinfo {year} {2024})}\BibitemShut {NoStop}%
\bibitem [{\citenamefont {Jisha}\ and\ \citenamefont
  {Prakash}(2024)}]{Jisha_2024}%
  \BibitemOpen
  \bibfield  {author} {\bibinfo {author} {\bibfnamefont {C}~\bibnamefont
  {Jisha}}\ and\ \bibinfo {author} {\bibfnamefont {Ravi}\ \bibnamefont
  {Prakash}},\ }\bibfield  {title} {\enquote {\bibinfo {title} {Universality of
  spectral fluctuations in open quantum chaotic systems},}\ }\href {\doibase
  10.1209/0295-5075/ad2c35} {\bibfield  {journal} {\bibinfo  {journal}
  {Europhysics Letters}\ }\textbf {\bibinfo {volume} {146}},\ \bibinfo {pages}
  {11001} (\bibinfo {year} {2024})}\BibitemShut {NoStop}%
\bibitem [{\citenamefont {Mishra}\ and\ \citenamefont
  {Sahoo}(2025)}]{MISHRA2025131000}%
  \BibitemOpen
  \bibfield  {author} {\bibinfo {author} {\bibfnamefont {Smitarani}\
  \bibnamefont {Mishra}}\ and\ \bibinfo {author} {\bibfnamefont {Shaon}\
  \bibnamefont {Sahoo}},\ }\bibfield  {title} {\enquote {\bibinfo {title}
  {Quantum thermalization and average entropy of a subsystem},}\ }\href
  {\doibase https://doi.org/10.1016/j.physleta.2025.131000} {\bibfield
  {journal} {\bibinfo  {journal} {Physics Letters A}\ }\textbf {\bibinfo
  {volume} {561}},\ \bibinfo {pages} {131000} (\bibinfo {year}
  {2025})}\BibitemShut {NoStop}%
\bibitem [{\citenamefont {Aravinda}\ \emph {et~al.}(2024)\citenamefont
  {Aravinda}, \citenamefont {Banerjee},\ and\ \citenamefont
  {Modak}}]{PhysRevA.110.042607}%
  \BibitemOpen
  \bibfield  {author} {\bibinfo {author} {\bibfnamefont {S.}~\bibnamefont
  {Aravinda}}, \bibinfo {author} {\bibfnamefont {Shilpak}\ \bibnamefont
  {Banerjee}}, \ and\ \bibinfo {author} {\bibfnamefont {Ranjan}\ \bibnamefont
  {Modak}},\ }\bibfield  {title} {\enquote {\bibinfo {title} {Ergodic and
  mixing quantum channels: From two-qubit to many-body quantum systems},}\
  }\href {\doibase 10.1103/PhysRevA.110.042607} {\bibfield  {journal} {\bibinfo
   {journal} {Phys. Rev. A}\ }\textbf {\bibinfo {volume} {110}},\ \bibinfo
  {pages} {042607} (\bibinfo {year} {2024})}\BibitemShut {NoStop}%
\bibitem [{\citenamefont {Burgarth}\ \emph {et~al.}(2013)\citenamefont
  {Burgarth}, \citenamefont {Chiribella}, \citenamefont {Giovannetti},
  \citenamefont {Perinotti},\ and\ \citenamefont {Yuasa}}]{Burgarth2013}%
  \BibitemOpen
  \bibfield  {author} {\bibinfo {author} {\bibfnamefont {D.}~\bibnamefont
  {Burgarth}}, \bibinfo {author} {\bibfnamefont {G.}~\bibnamefont
  {Chiribella}}, \bibinfo {author} {\bibfnamefont {V.}~\bibnamefont
  {Giovannetti}}, \bibinfo {author} {\bibfnamefont {P.}~\bibnamefont
  {Perinotti}}, \ and\ \bibinfo {author} {\bibfnamefont {K.}~\bibnamefont
  {Yuasa}},\ }\bibfield  {title} {\enquote {\bibinfo {title} {Ergodic and
  mixing quantum channels in finite dimensions},}\ }\href@noop {} {\bibfield
  {journal} {\bibinfo  {journal} {New Journal of Physics}\ }\textbf {\bibinfo
  {volume} {15}},\ \bibinfo {pages} {073045} (\bibinfo {year}
  {2013})}\BibitemShut {NoStop}%
\bibitem [{\citenamefont {Movassagh}\ and\ \citenamefont
  {Schenker}(2022)}]{Movassagh2022}%
  \BibitemOpen
  \bibfield  {author} {\bibinfo {author} {\bibfnamefont {R.}~\bibnamefont
  {Movassagh}}\ and\ \bibinfo {author} {\bibfnamefont {J.}~\bibnamefont
  {Schenker}},\ }\bibfield  {title} {\enquote {\bibinfo {title} {An ergodic
  theorem for quantum processes with applications to matrix product states},}\
  }\href@noop {} {\bibfield  {journal} {\bibinfo  {journal} {Communications in
  Mathematical Physics}\ }\textbf {\bibinfo {volume} {395}},\ \bibinfo {pages}
  {1175} (\bibinfo {year} {2022})}\BibitemShut {NoStop}%
\bibitem [{\citenamefont {Movassagh}\ and\ \citenamefont
  {Schenker}(2021)}]{Movassagh2021}%
  \BibitemOpen
  \bibfield  {author} {\bibinfo {author} {\bibfnamefont {R.}~\bibnamefont
  {Movassagh}}\ and\ \bibinfo {author} {\bibfnamefont {J.}~\bibnamefont
  {Schenker}},\ }\bibfield  {title} {\enquote {\bibinfo {title} {Theory of
  ergodic quantum processes},}\ }\href@noop {} {\bibfield  {journal} {\bibinfo
  {journal} {Physical Review X}\ }\textbf {\bibinfo {volume} {11}},\ \bibinfo
  {pages} {041001} (\bibinfo {year} {2021})}\BibitemShut {NoStop}%
\bibitem [{\citenamefont {Zanardi}\ \emph
  {et~al.}(2000{\natexlab{b}})\citenamefont {Zanardi}, \citenamefont {Zalka},\
  and\ \citenamefont {Faoro}}]{Zanardi2000}%
  \BibitemOpen
  \bibfield  {author} {\bibinfo {author} {\bibfnamefont {Paolo}\ \bibnamefont
  {Zanardi}}, \bibinfo {author} {\bibfnamefont {Christof}\ \bibnamefont
  {Zalka}}, \ and\ \bibinfo {author} {\bibfnamefont {Lara}\ \bibnamefont
  {Faoro}},\ }\bibfield  {title} {\enquote {\bibinfo {title} {Entangling power
  of quantum evolutions},}\ }\href {\doibase 10.1103/PhysRevA.62.030301}
  {\bibfield  {journal} {\bibinfo  {journal} {Phys. Rev. A}\ }\textbf {\bibinfo
  {volume} {62}},\ \bibinfo {pages} {030301} (\bibinfo {year}
  {2000}{\natexlab{b}})}\BibitemShut {NoStop}%
\bibitem [{\citenamefont {Kraus}\ and\ \citenamefont
  {Cirac}(2001)}]{Kraus2001}%
  \BibitemOpen
  \bibfield  {author} {\bibinfo {author} {\bibfnamefont {Barbara}\ \bibnamefont
  {Kraus}}\ and\ \bibinfo {author} {\bibfnamefont {Juan~Ignacio}\ \bibnamefont
  {Cirac}},\ }\bibfield  {title} {\enquote {\bibinfo {title} {Optimal creation
  of entanglement using a two-qubit gate},}\ }\href {\doibase
  10.1103/PhysRevA.63.062309} {\bibfield  {journal} {\bibinfo  {journal}
  {Physical Review A}\ }\textbf {\bibinfo {volume} {63}},\ \bibinfo {pages}
  {062309} (\bibinfo {year} {2001})}\BibitemShut {NoStop}%
\bibitem [{\citenamefont {Makhlin}(2002)}]{Makhlin2002}%
  \BibitemOpen
  \bibfield  {author} {\bibinfo {author} {\bibfnamefont {Yuriy}\ \bibnamefont
  {Makhlin}},\ }\bibfield  {title} {\enquote {\bibinfo {title} {Nonlocal
  properties of two-qubit gates and mixed states, and the optimization of
  quantum computations},}\ }\href {\doibase 10.1023/A:1022144002391} {\bibfield
   {journal} {\bibinfo  {journal} {Quantum Information Processing}\ }\textbf
  {\bibinfo {volume} {1}},\ \bibinfo {pages} {243--252} (\bibinfo {year}
  {2002})}\BibitemShut {NoStop}%
\bibitem [{\citenamefont {Zhang}\ \emph {et~al.}(2003)\citenamefont {Zhang},
  \citenamefont {Vala}, \citenamefont {Sastry},\ and\ \citenamefont
  {Whaley}}]{Zhang2003}%
  \BibitemOpen
  \bibfield  {author} {\bibinfo {author} {\bibfnamefont {Jun}\ \bibnamefont
  {Zhang}}, \bibinfo {author} {\bibfnamefont {Jiri}\ \bibnamefont {Vala}},
  \bibinfo {author} {\bibfnamefont {Shankar}\ \bibnamefont {Sastry}}, \ and\
  \bibinfo {author} {\bibfnamefont {K.~Birgitta}\ \bibnamefont {Whaley}},\
  }\bibfield  {title} {\enquote {\bibinfo {title} {Geometric theory of nonlocal
  two-qubit operations},}\ }\href {\doibase 10.1103/PhysRevA.67.042313}
  {\bibfield  {journal} {\bibinfo  {journal} {Phys. Rev. A}\ }\textbf {\bibinfo
  {volume} {67}},\ \bibinfo {pages} {042313} (\bibinfo {year}
  {2003})}\BibitemShut {NoStop}%
\bibitem [{\citenamefont {Jonnadula}\ \emph {et~al.}(2017)\citenamefont
  {Jonnadula}, \citenamefont {Mandayam}, \citenamefont {{\.Z}yczkowski},\ and\
  \citenamefont {Lakshminarayan}}]{Jonnadula2017}%
  \BibitemOpen
  \bibfield  {author} {\bibinfo {author} {\bibfnamefont {Bhargavi}\
  \bibnamefont {Jonnadula}}, \bibinfo {author} {\bibfnamefont {Prabha}\
  \bibnamefont {Mandayam}}, \bibinfo {author} {\bibfnamefont {Karol}\
  \bibnamefont {{\.Z}yczkowski}}, \ and\ \bibinfo {author} {\bibfnamefont
  {Arul}\ \bibnamefont {Lakshminarayan}},\ }\bibfield  {title} {\enquote
  {\bibinfo {title} {Impact of local dynamics on entangling power},}\ }\href
  {\doibase 10.1103/PhysRevA.95.040302} {\bibfield  {journal} {\bibinfo
  {journal} {Physical Review A}\ }\textbf {\bibinfo {volume} {95}},\ \bibinfo
  {pages} {040302(R)} (\bibinfo {year} {2017})}\BibitemShut {NoStop}%
\bibitem [{\citenamefont {Manna}\ \emph {et~al.}(2024)\citenamefont {Manna},
  \citenamefont {Madhok},\ and\ \citenamefont {Lakshminarayan}}]{Manna2024}%
  \BibitemOpen
  \bibfield  {author} {\bibinfo {author} {\bibfnamefont {Sourav}\ \bibnamefont
  {Manna}}, \bibinfo {author} {\bibfnamefont {Vaibhav}\ \bibnamefont {Madhok}},
  \ and\ \bibinfo {author} {\bibfnamefont {Arul}\ \bibnamefont
  {Lakshminarayan}},\ }\bibfield  {title} {\enquote {\bibinfo {title}
  {Entangling power, gate typicality, and measurement-induced phase
  transitions},}\ }\href {\doibase 10.1103/PhysRevA.110.062422} {\bibfield
  {journal} {\bibinfo  {journal} {Physical Review A}\ }\textbf {\bibinfo
  {volume} {110}},\ \bibinfo {pages} {062422} (\bibinfo {year}
  {2024})}\BibitemShut {NoStop}%
\bibitem [{\citenamefont {Qiu}\ \emph {et~al.}(2025)\citenamefont {Qiu},
  \citenamefont {Song},\ and\ \citenamefont {Chen}}]{Qiu2025}%
  \BibitemOpen
  \bibfield  {author} {\bibinfo {author} {\bibfnamefont {Xinyu}\ \bibnamefont
  {Qiu}}, \bibinfo {author} {\bibfnamefont {Zhiwei}\ \bibnamefont {Song}}, \
  and\ \bibinfo {author} {\bibfnamefont {Lin}\ \bibnamefont {Chen}},\
  }\bibfield  {title} {\enquote {\bibinfo {title} {Multipartite entangling
  power by von neumann entropy},}\ }\href {\doibase
  10.1103/PhysRevA.111.022407} {\bibfield  {journal} {\bibinfo  {journal}
  {Physical Review A}\ }\textbf {\bibinfo {volume} {111}},\ \bibinfo {pages}
  {022407} (\bibinfo {year} {2025})}\BibitemShut {NoStop}%
\bibitem [{\citenamefont {Malik}\ \emph {et~al.}(2026)\citenamefont {Malik},
  \citenamefont {Shukla}, \citenamefont {Joshi}, \citenamefont {Aravinda},\
  and\ \citenamefont {Mishra}}]{Malik2026}%
  \BibitemOpen
  \bibfield  {author} {\bibinfo {author} {\bibfnamefont {Gaurav~Rudra}\
  \bibnamefont {Malik}}, \bibinfo {author} {\bibfnamefont {Rohit~Kumar}\
  \bibnamefont {Shukla}}, \bibinfo {author} {\bibfnamefont {Sudhanva}\
  \bibnamefont {Joshi}}, \bibinfo {author} {\bibfnamefont {S.}~\bibnamefont
  {Aravinda}}, \ and\ \bibinfo {author} {\bibfnamefont {Sunil~Kumar}\
  \bibnamefont {Mishra}},\ }\bibfield  {title} {\enquote {\bibinfo {title}
  {Entanglement structure for a finite system under dual-unitary dynamics},}\
  }\href@noop {} {\bibfield  {journal} {\bibinfo  {journal} {Physical Review
  B}\ }\textbf {\bibinfo {volume} {113}},\ \bibinfo {pages} {064307} (\bibinfo
  {year} {2026})}\BibitemShut {NoStop}%
\bibitem [{\citenamefont {Bansal}\ \emph {et~al.}(2025)\citenamefont {Bansal},
  \citenamefont {Mok}, \citenamefont {Bharti}, \citenamefont {Koh},\ and\
  \citenamefont {Haug}}]{PRXQuantum.6.020322}%
  \BibitemOpen
  \bibfield  {author} {\bibinfo {author} {\bibfnamefont {Nikhil}\ \bibnamefont
  {Bansal}}, \bibinfo {author} {\bibfnamefont {Wai-Keong}\ \bibnamefont {Mok}},
  \bibinfo {author} {\bibfnamefont {Kishor}\ \bibnamefont {Bharti}}, \bibinfo
  {author} {\bibfnamefont {Dax~Enshan}\ \bibnamefont {Koh}}, \ and\ \bibinfo
  {author} {\bibfnamefont {Tobias}\ \bibnamefont {Haug}},\ }\bibfield  {title}
  {\enquote {\bibinfo {title} {Pseudorandom density matrices},}\ }\href
  {\doibase 10.1103/PRXQuantum.6.020322} {\bibfield  {journal} {\bibinfo
  {journal} {PRX Quantum}\ }\textbf {\bibinfo {volume} {6}},\ \bibinfo {pages}
  {020322} (\bibinfo {year} {2025})}\BibitemShut {NoStop}%
\bibitem [{\citenamefont {Song}\ \emph {et~al.}(2025)\citenamefont {Song},
  \citenamefont {Yang}, \citenamefont {Liu}, \citenamefont {Zhang},
  \citenamefont {Xue}, \citenamefont {Mi}, \citenamefont {Zhang}, \citenamefont
  {Yan}, \citenamefont {Jin},\ and\ \citenamefont {Yu}}]{npr7-b7kq}%
  \BibitemOpen
  \bibfield  {author} {\bibinfo {author} {\bibfnamefont {Juan}\ \bibnamefont
  {Song}}, \bibinfo {author} {\bibfnamefont {Shuang}\ \bibnamefont {Yang}},
  \bibinfo {author} {\bibfnamefont {Pei}\ \bibnamefont {Liu}}, \bibinfo
  {author} {\bibfnamefont {Hui-Li}\ \bibnamefont {Zhang}}, \bibinfo {author}
  {\bibfnamefont {Guang-Ming}\ \bibnamefont {Xue}}, \bibinfo {author}
  {\bibfnamefont {Zhen-Yu}\ \bibnamefont {Mi}}, \bibinfo {author}
  {\bibfnamefont {Wen-Gang}\ \bibnamefont {Zhang}}, \bibinfo {author}
  {\bibfnamefont {Fei}\ \bibnamefont {Yan}}, \bibinfo {author} {\bibfnamefont
  {Yi-Rong}\ \bibnamefont {Jin}}, \ and\ \bibinfo {author} {\bibfnamefont
  {Hai-Feng}\ \bibnamefont {Yu}},\ }\bibfield  {title} {\enquote {\bibinfo
  {title} {Realization of high-fidelity perfect entanglers between remote
  superconducting quantum processors},}\ }\href {\doibase 10.1103/npr7-b7kq}
  {\bibfield  {journal} {\bibinfo  {journal} {Phys. Rev. Lett.}\ }\textbf
  {\bibinfo {volume} {135}},\ \bibinfo {pages} {050603} (\bibinfo {year}
  {2025})}\BibitemShut {NoStop}%
\bibitem [{\citenamefont {Vijaywargia}\ and\ \citenamefont
  {Lakshminarayan}(2025)}]{PhysRevE.111.014210}%
  \BibitemOpen
  \bibfield  {author} {\bibinfo {author} {\bibfnamefont {Bidhi}\ \bibnamefont
  {Vijaywargia}}\ and\ \bibinfo {author} {\bibfnamefont {Arul}\ \bibnamefont
  {Lakshminarayan}},\ }\bibfield  {title} {\enquote {\bibinfo {title}
  {Quantum-classical correspondence in quantum channels},}\ }\href {\doibase
  10.1103/PhysRevE.111.014210} {\bibfield  {journal} {\bibinfo  {journal}
  {Phys. Rev. E}\ }\textbf {\bibinfo {volume} {111}},\ \bibinfo {pages}
  {014210} (\bibinfo {year} {2025})}\BibitemShut {NoStop}%
\bibitem [{\citenamefont {Zhang}(2025)}]{d8kg-h1t7}%
  \BibitemOpen
  \bibfield  {author} {\bibinfo {author} {\bibfnamefont {Lin}\ \bibnamefont
  {Zhang}},\ }\bibfield  {title} {\enquote {\bibinfo {title} {Operator
  entanglement in su(2)-symmetric dissipative quantum many-body dynamics},}\
  }\href {\doibase 10.1103/d8kg-h1t7} {\bibfield  {journal} {\bibinfo
  {journal} {Phys. Rev. B}\ }\textbf {\bibinfo {volume} {112}},\ \bibinfo
  {pages} {144303} (\bibinfo {year} {2025})}\BibitemShut {NoStop}%
\bibitem [{\citenamefont {Mondal}\ \emph {et~al.}(2025)\citenamefont {Mondal},
  \citenamefont {Hazra},\ and\ \citenamefont {Sen(De)}}]{jl2b-bxfn}%
  \BibitemOpen
  \bibfield  {author} {\bibinfo {author} {\bibfnamefont {Sudipta}\ \bibnamefont
  {Mondal}}, \bibinfo {author} {\bibfnamefont {Samir~Kumar}\ \bibnamefont
  {Hazra}}, \ and\ \bibinfo {author} {\bibfnamefont {Aditi}\ \bibnamefont
  {Sen(De)}},\ }\bibfield  {title} {\enquote {\bibinfo {title} {Imperfect
  entangling power of quantum gates},}\ }\href {\doibase 10.1103/jl2b-bxfn}
  {\bibfield  {journal} {\bibinfo  {journal} {Phys. Rev. A}\ }\textbf {\bibinfo
  {volume} {112}},\ \bibinfo {pages} {012417} (\bibinfo {year}
  {2025})}\BibitemShut {NoStop}%
\bibitem [{\citenamefont {Mandarino}\ \emph {et~al.}(2018)\citenamefont
  {Mandarino}, \citenamefont {Linowski},\ and\ \citenamefont
  {{\.Z}yczkowski}}]{Mandarino2018}%
  \BibitemOpen
  \bibfield  {author} {\bibinfo {author} {\bibfnamefont {Antonio}\ \bibnamefont
  {Mandarino}}, \bibinfo {author} {\bibfnamefont {Tomasz}\ \bibnamefont
  {Linowski}}, \ and\ \bibinfo {author} {\bibfnamefont {Karol}\ \bibnamefont
  {{\.Z}yczkowski}},\ }\bibfield  {title} {\enquote {\bibinfo {title}
  {Bipartite unitary gates and billiard dynamics in the {W}eyl chamber},}\
  }\href {\doibase 10.1103/PhysRevA.98.012335} {\bibfield  {journal} {\bibinfo
  {journal} {Phys. Rev. A}\ }\textbf {\bibinfo {volume} {98}},\ \bibinfo
  {pages} {012335} (\bibinfo {year} {2018})}\BibitemShut {NoStop}%
\bibitem [{\citenamefont {Goli}\ \emph {et~al.}(2013)\citenamefont {Goli},
  \citenamefont {Sahoo}, \citenamefont {Ramasesha},\ and\ \citenamefont
  {Sen}}]{Sahoo2013}%
  \BibitemOpen
  \bibfield  {author} {\bibinfo {author} {\bibfnamefont {V.~M. L.
  Durga~Prasad}\ \bibnamefont {Goli}}, \bibinfo {author} {\bibfnamefont
  {Shaon}\ \bibnamefont {Sahoo}}, \bibinfo {author} {\bibfnamefont
  {S.}~\bibnamefont {Ramasesha}}, \ and\ \bibinfo {author} {\bibfnamefont
  {Diptiman}\ \bibnamefont {Sen}},\ }\bibfield  {title} {\enquote {\bibinfo
  {title} {Quantum phases of dimerized and frustrated heisenberg spin chains
  with s = 1/2, 1 and 3/2: An entanglement entropy and fidelity study},}\
  }\href {\doibase 10.1088/0953-8984/25/12/125603} {\bibfield  {journal}
  {\bibinfo  {journal} {Journal of Physics: Condensed Matter}\ }\textbf
  {\bibinfo {volume} {25}},\ \bibinfo {pages} {125603} (\bibinfo {year}
  {2013})}\BibitemShut {NoStop}%
\bibitem [{\citenamefont {Sahoo}\ \emph {et~al.}(2014)\citenamefont {Sahoo},
  \citenamefont {Goli}, \citenamefont {Sen},\ and\ \citenamefont
  {Ramasesha}}]{Sahoo2014}%
  \BibitemOpen
  \bibfield  {author} {\bibinfo {author} {\bibfnamefont {Shaon}\ \bibnamefont
  {Sahoo}}, \bibinfo {author} {\bibfnamefont {V.~M. L. Durga~Prasad}\
  \bibnamefont {Goli}}, \bibinfo {author} {\bibfnamefont {Diptiman}\
  \bibnamefont {Sen}}, \ and\ \bibinfo {author} {\bibfnamefont
  {S.}~\bibnamefont {Ramasesha}},\ }\bibfield  {title} {\enquote {\bibinfo
  {title} {Studies on a frustrated heisenberg spin chain with alternating
  ferromagnetic and antiferromagnetic exchanges},}\ }\href {\doibase
  10.1088/0953-8984/26/27/276002} {\bibfield  {journal} {\bibinfo  {journal}
  {Journal of Physics: Condensed Matter}\ }\textbf {\bibinfo {volume} {26}},\
  \bibinfo {pages} {276002} (\bibinfo {year} {2014})}\BibitemShut {NoStop}%
\bibitem [{\citenamefont {Sahoo}\ \emph {et~al.}(2020)\citenamefont {Sahoo},
  \citenamefont {Dey}, \citenamefont {Saha},\ and\ \citenamefont
  {Kumar}}]{Sahoo2020}%
  \BibitemOpen
  \bibfield  {author} {\bibinfo {author} {\bibfnamefont {Shaon}\ \bibnamefont
  {Sahoo}}, \bibinfo {author} {\bibfnamefont {D.}~\bibnamefont {Dey}}, \bibinfo
  {author} {\bibfnamefont {S.}~\bibnamefont {Saha}}, \ and\ \bibinfo {author}
  {\bibfnamefont {M.}~\bibnamefont {Kumar}},\ }\bibfield  {title} {\enquote
  {\bibinfo {title} {Haldane and dimer phases in a frustrated spin chain: an
  exact ground state and associated topological phase transition},}\ }\href
  {\doibase 10.1088/1361-648X/ab8b75} {\bibfield  {journal} {\bibinfo
  {journal} {Journal of Physics: Condensed Matter}\ }\textbf {\bibinfo {volume}
  {32}},\ \bibinfo {pages} {335601} (\bibinfo {year} {2020})}\BibitemShut
  {NoStop}%
\bibitem [{\citenamefont {Zanardi}(2001{\natexlab{b}})}]{Zanardi2001}%
  \BibitemOpen
  \bibfield  {author} {\bibinfo {author} {\bibfnamefont {Paolo}\ \bibnamefont
  {Zanardi}},\ }\bibfield  {title} {\enquote {\bibinfo {title} {Entanglement of
  quantum evolutions},}\ }\href {\doibase 10.1103/PhysRevA.63.040304}
  {\bibfield  {journal} {\bibinfo  {journal} {Phys. Rev. A}\ }\textbf {\bibinfo
  {volume} {63}},\ \bibinfo {pages} {040304} (\bibinfo {year}
  {2001}{\natexlab{b}})}\BibitemShut {NoStop}%
\bibitem [{\citenamefont {Wang}\ \emph {et~al.}(2003)\citenamefont {Wang},
  \citenamefont {Sanders},\ and\ \citenamefont {Berry}}]{Wang2003}%
  \BibitemOpen
  \bibfield  {author} {\bibinfo {author} {\bibfnamefont {Xiaoguang}\
  \bibnamefont {Wang}}, \bibinfo {author} {\bibfnamefont {Barry~C.}\
  \bibnamefont {Sanders}}, \ and\ \bibinfo {author} {\bibfnamefont
  {Dominic~W.}\ \bibnamefont {Berry}},\ }\bibfield  {title} {\enquote {\bibinfo
  {title} {Entangling power and operator entanglement in qudit systems},}\
  }\href {\doibase 10.1103/PhysRevA.67.042323} {\bibfield  {journal} {\bibinfo
  {journal} {Phys. Rev. A}\ }\textbf {\bibinfo {volume} {67}},\ \bibinfo
  {pages} {042323} (\bibinfo {year} {2003})}\BibitemShut {NoStop}%
\bibitem [{\citenamefont {Lubkin}(1978)}]{Lubkin1978}%
  \BibitemOpen
  \bibfield  {author} {\bibinfo {author} {\bibfnamefont {Elihu}\ \bibnamefont
  {Lubkin}},\ }\bibfield  {title} {\enquote {\bibinfo {title} {Entropy of an
  n-system from its correlation with a k-reservoir},}\ }\href {\doibase
  10.1063/1.523763} {\bibfield  {journal} {\bibinfo  {journal} {Journal of
  Mathematical Physics}\ }\textbf {\bibinfo {volume} {19}},\ \bibinfo {pages}
  {1028--1031} (\bibinfo {year} {1978})}\BibitemShut {NoStop}%
\bibitem [{\citenamefont {Balakrishnan}\ and\ \citenamefont
  {Sankaranarayanan}(2010)}]{Balakrishnan2010}%
  \BibitemOpen
  \bibfield  {author} {\bibinfo {author} {\bibfnamefont {S.}~\bibnamefont
  {Balakrishnan}}\ and\ \bibinfo {author} {\bibfnamefont {R.}~\bibnamefont
  {Sankaranarayanan}},\ }\bibfield  {title} {\enquote {\bibinfo {title}
  {Entangling power and local invariants of two-qubit gates},}\ }\href
  {\doibase 10.1103/PhysRevA.82.034301} {\bibfield  {journal} {\bibinfo
  {journal} {Physical Review A}\ }\textbf {\bibinfo {volume} {82}},\ \bibinfo
  {pages} {034301} (\bibinfo {year} {2010})}\BibitemShut {NoStop}%
\bibitem [{\citenamefont {Lakshminarayan}\ and\ \citenamefont
  {Subrahmanyam}(2005)}]{arul.2005}%
  \BibitemOpen
  \bibfield  {author} {\bibinfo {author} {\bibfnamefont {Arul}\ \bibnamefont
  {Lakshminarayan}}\ and\ \bibinfo {author} {\bibfnamefont {V.}~\bibnamefont
  {Subrahmanyam}},\ }\bibfield  {title} {\enquote {\bibinfo {title}
  {Multipartite entanglement in a one-dimensional time-dependent ising
  model},}\ }\href {\doibase 10.1103/PhysRevA.71.062334} {\bibfield  {journal}
  {\bibinfo  {journal} {Phys. Rev. A}\ }\textbf {\bibinfo {volume} {71}},\
  \bibinfo {pages} {062334} (\bibinfo {year} {2005})}\BibitemShut {NoStop}%
\bibitem [{\citenamefont {Gross}\ \emph {et~al.}(2007)\citenamefont {Gross},
  \citenamefont {Audenaert},\ and\ \citenamefont
  {Eisert}}]{Gross2007UnitaryDesigns}%
  \BibitemOpen
  \bibfield  {author} {\bibinfo {author} {\bibfnamefont {D.}~\bibnamefont
  {Gross}}, \bibinfo {author} {\bibfnamefont {K.}~\bibnamefont {Audenaert}}, \
  and\ \bibinfo {author} {\bibfnamefont {J.}~\bibnamefont {Eisert}},\
  }\bibfield  {title} {\enquote {\bibinfo {title} {Evenly distributed
  unitaries: On the structure of unitary designs},}\ }\href {\doibase
  10.1063/1.2716992} {\bibfield  {journal} {\bibinfo  {journal} {Journal of
  Mathematical Physics}\ }\textbf {\bibinfo {volume} {48}},\ \bibinfo {pages}
  {052104} (\bibinfo {year} {2007})}\BibitemShut {NoStop}%
\bibitem [{\citenamefont {Suzuki}\ \emph {et~al.}(2026)\citenamefont {Suzuki},
  \citenamefont {Katsura}, \citenamefont {Mitsuhashi}, \citenamefont {Soejima},
  \citenamefont {Eisert},\ and\ \citenamefont {Yoshioka}}]{9tnv-2559}%
  \BibitemOpen
  \bibfield  {author} {\bibinfo {author} {\bibfnamefont {Ryotaro}\ \bibnamefont
  {Suzuki}}, \bibinfo {author} {\bibfnamefont {Hosho}\ \bibnamefont {Katsura}},
  \bibinfo {author} {\bibfnamefont {Yosuke}\ \bibnamefont {Mitsuhashi}},
  \bibinfo {author} {\bibfnamefont {Tomohiro}\ \bibnamefont {Soejima}},
  \bibinfo {author} {\bibfnamefont {Jens}\ \bibnamefont {Eisert}}, \ and\
  \bibinfo {author} {\bibfnamefont {Nobuyuki}\ \bibnamefont {Yoshioka}},\
  }\bibfield  {title} {\enquote {\bibinfo {title} {More global randomness from
  less-random local gates},}\ }\href {\doibase 10.1103/9tnv-2559} {\bibfield
  {journal} {\bibinfo  {journal} {Phys. Rev. A}\ }\textbf {\bibinfo {volume}
  {114}},\ \bibinfo {pages} {022424} (\bibinfo {year} {2026})}\BibitemShut
  {NoStop}%
\bibitem [{\citenamefont {Rampp}\ \emph {et~al.}(2025)\citenamefont {Rampp},
  \citenamefont {Rather},\ and\ \citenamefont
  {Claeys}}]{rampp2025solvablequantumcircuitsspacetime}%
  \BibitemOpen
  \bibfield  {author} {\bibinfo {author} {\bibfnamefont {Michael~A.}\
  \bibnamefont {Rampp}}, \bibinfo {author} {\bibfnamefont {Suhail~A.}\
  \bibnamefont {Rather}}, \ and\ \bibinfo {author} {\bibfnamefont {Pieter~W.}\
  \bibnamefont {Claeys}},\ }\href {https://arxiv.org/abs/2512.15871} {\enquote
  {\bibinfo {title} {Solvable quantum circuits from spacetime lattices},}\ }
  (\bibinfo {year} {2025}),\ \Eprint {http://arxiv.org/abs/2512.15871}
  {arXiv:2512.15871 [quant-ph]} \BibitemShut {NoStop}%
\bibitem [{\citenamefont {Rampp}\ \emph {et~al.}(2026)\citenamefont {Rampp},
  \citenamefont {Rather},\ and\ \citenamefont
  {Claeys}}]{rampp2026infinitelevelhierarchysolvablequantum}%
  \BibitemOpen
  \bibfield  {author} {\bibinfo {author} {\bibfnamefont {Michael~A.}\
  \bibnamefont {Rampp}}, \bibinfo {author} {\bibfnamefont {Suhail~A.}\
  \bibnamefont {Rather}}, \ and\ \bibinfo {author} {\bibfnamefont {Pieter~W.}\
  \bibnamefont {Claeys}},\ }\href {https://arxiv.org/abs/2606.23803} {\enquote
  {\bibinfo {title} {Infinite-level hierarchy of solvable quantum circuits},}\
  } (\bibinfo {year} {2026}),\ \Eprint {http://arxiv.org/abs/2606.23803}
  {arXiv:2606.23803 [quant-ph]} \BibitemShut {NoStop}%
\end{thebibliography}

%
\end{document}